\documentclass[lettersize,journal]{IEEEtran}
\usepackage{amsmath,amsfonts}
\usepackage{array}
\usepackage[caption=false,font=normalsize,labelfont=sf,textfont=sf]{subfig}
\usepackage{textcomp}
\usepackage{stfloats}
\usepackage{url}
\usepackage{verbatim}
\usepackage{graphicx}
\usepackage{cite}
\usepackage{orcidlink}
\usepackage{enumitem}
\usepackage{algorithm}
\usepackage{algpseudocode}  
\usepackage{caption}
\usepackage{subfig}
\usepackage{amssymb}    
\usepackage{amsthm}     
\usepackage[T1]{fontenc}

\algrenewcommand\algorithmicrequire{\textbf{Input:}}
\algrenewcommand\algorithmicensure{\textbf{Output:}}

\def\BibTeX{{\rm B\kern-.05em{\sc i\kern-.025em b}\kern-.08em
		T\kern-.1667em\lower.7ex\hbox{E}\kern-.125emX}}
\usepackage{balance}
\makeatletter
\newcommand*{\rom}[1]{\expandafter\@slowromancap\romannumeral #1@}
\makeatother
\providecommand{\triangleq}{\stackrel{\triangle}{=}}

\theoremstyle{remark}

\begin{document}
	\title{Joint Deployment and Resource Management for Group-Tag-Based Physical-Layer Authentication and Collaborative Communication in NTN}
	\author{Parsa~Rajabi~\orcidlink{0009-0005-4645-9444},
		Mohammad~Reza~Abedi~\orcidlink{0000-0003-4114-1339}~\IEEEmembership{Member,~IEEE},
		Paeiz~Azmi~\orcidlink{0000-0001-9736-3462}~\IEEEmembership{Senior~Member,~IEEE},
		Nader~Mokari~\orcidlink{0000-0001-5364-8888}~\IEEEmembership{Senior~Member,~IEEE}
		
		\thanks{This work has been submitted to the IEEE for possible publication.
			Copyright may be transferred without notice, after which this version may
			no longer be accessible.}
		
		\thanks{P. Rajabi, M. R. Abedi, P. Azmi, and N. Mokari are with the Department of Electrical and Computer Engineering, Tarbiat Modares University, Tehran, 14115-111, Iran (e-mail: parsa\_rajabi@modares.ac.ir; Mohammadreza\_abedi@modares.ac.ir; pazmi@modares.ac.ir; nader.mokari@modares.ac.ir).}
	}

	\maketitle
	
	\begin{abstract}
		Physical-layer authentication (PLA) in non-terrestrial networks (NTNs) has primarily been studied from a security perspective, while its system-level integration with communication remains less explored. 
		We propose the secure adaptive framework for authentication
		in multi-zone networks (SAFA-MZ), a system-level formulation for the joint management of collaborative group-tag and private-data delivery in NTN, where high-altitude platform station (HAPS) placement, user association, and transmission powers are jointly selected under service and secrecy constraints.
		Group-level authentication tags and user data are collaboratively transmitted by multiple HAPSs, with artificial noise (AN) mitigating information leakage. We jointly optimize HAPS placement, user association, and power allocation to balance deployment regularity, tag distribution, and user-rate fairness under secrecy, quality-of-service (QoS), power, and coverage constraints. The problem is addressed using a repair-based cross-entropy method (RCEM) and a graph-aware advantage actor-critic (A2C) algorithm (GA2C). Under colluding and non-colluding eavesdropper models, GA2C improves average secrecy spectral efficiency (SSE) by 135--139\% over its single-connect baseline and 21.1--23.2\% over its baseline without AN. Across ten independent random seeds, GA2C achieves approximately
		23.4\% and 17.9\% higher mean per-user SSE under non-colluding
		and colluding eavesdropping, respectively, and a 16.2\% higher
		mean private-rate Jain's index than RCEM. These results quantify the gains of multi-node connectivity
		and AN within the proposed framework under the evaluated assumptions.
	\end{abstract}
	
	
	\section{Introduction}
	
	\IEEEPARstart{N}{on-terrestrial} networks (NTNs) integrate satellites, high-altitude platform stations (HAPSs), and unmanned aerial vehicles (UAVs) with terrestrial infrastructure. Long links, mobility, Doppler effects, and changing topologies introduce security challenges. Physical-layer authentication (PLA) complements cryptographic authentication through passive fingerprints or actively embedded tags~\cite{ref8,ref85,ref54}.
	
	Passive PLA exploits channel state information (CSI) at low signal-to-noise ratios (SNRs)~\cite{ref42} and beam-pattern fingerprints in millimeter-wave multiple-input multiple-output (MIMO) systems~\cite{ref43}. Covert authentication has also
	been studied~\cite{teng2025covertauth}. Active PLA includes challenge--response (CR) schemes~\cite{ref47}, message-aware and cancellation-based receivers~\cite{ref29,ref48}, and fading-aware tag power control~\cite{ref49}. A reconfigurable intelligent surface (RIS) supports tag construction and CR authentication~\cite{ref50,ref30}, including randomized configurations at a receiving base station~\cite{ref82}; phase-tag schemes use key-dependent components~\cite{ref76}.
	
	Security and privacy in networks have been reviewed in \cite{khan2024security}, while PLA-based blockchain consensus has been studied in \cite{venkatesan2024lightweight}.
	Surveys cover machine learning (ML)-based and tag-based PLA under time-varying channels~\cite{ref53}. Deep learning (DL) supports pilot-contamination defense~\cite{ref85} and satellite authentication~\cite{ref26}. We adopt active PLA for its controllability, since line-of-sight (LoS) dominance and mobility can weaken passive NTN fingerprints.
	
	Distributed PLA (DPLA) improves robustness through multi-anchor fusion~\cite{ref74}. Prior work combines PLA with secure key exchange~\cite{zhang2025authscatter}, tag-based encoding and AN in massive-MIMO UAV communications~\cite{wang2025joint}, and covert transmission~\cite{teng2025covertauth}. System-level reliability and delay in multi-node networks are studied in~\cite{salehi2023reliability}, and visibility-aware association and resource allocation in multi-slice low Earth orbit (LEO) networks in~\cite{mahyoub2025visibility}. 
	The authors in \cite{zhou2025incentive} coordinate reliable PLA
	collaborators through Stackelberg incentives and quality-aware
	selection to improve participation and authentication accuracy.
	The authors in \cite{chen2018physical} develop information-theoretically
	secure message authentication using wiretap coding, with key reuse
	and asymptotic authentication-rate guarantees.
	The authors in \cite{xie2018security} jointly model authentication
	covertness, security, and robustness and derive detection and
	false-alarm probabilities over fading channels.
	The authors in \cite{liu2026use} expose correlation-based attacks
	on CR PLA and establish a randomness-testing
	guideline for its practical applicability.
	The authors in \cite{zhang2025enhanced} combine carrier frequency
	offset fingerprints and elliptic-curve session keys for two-way
	UAV authentication and data confidentiality.
	These studies leave the joint design of group-tag delivery, collaborative communication, and multi-HAPS resources insufficiently explored. Our work instead jointly optimizes multi-HAPS placement, user
	association, and power allocation for collaborative group-tag
	and private-data delivery under secrecy and service constraints.
	
	Graph-based learning supports distributed throughput optimization in cell-free massive MIMO~\cite{ref68}, fairness-aware satellite networks~\cite{ref69}, and heterogeneous-graph modeling for simultaneously transmitting and reflecting RIS (STAR-RIS)~\cite{ref70}. Graph neural networks (GNNs) encode NTN connectivity through message passing~\cite{ref67}; reinforcement learning (RL) supports sequential resource decisions under evolving topologies~\cite{ref33}.
	
	We propose the secure adaptive framework for authentication in multi-zone networks (SAFA-MZ), integrating group-level tag-based PLA with collaborative multi-node transmission. A group master key, securely provisioned during trusted initialization, generates the registered tag carried by the group-common stream. The study focuses on system-level integration; replay resistance, freshness, compromised nodes, and cryptographic key distribution are outside its scope. We jointly optimize HAPS placement, user association, and power allocation to balance layout regularity, group-association distribution, and private-rate fairness under secrecy, quality-of-service (QoS), coverage, capacity, power, and tag-delivery constraints. A repair-based cross-entropy method (RCEM) solves individual snapshots, while graph-aware advantage actor--critic (GA2C) learns a reusable policy. Its graph encoder is permutation-equivariant, and mean pooling is permutation-invariant. 
	The framework couples group-tag delivery and private-data
	service through shared placement, association, and power
	decisions. Simulations evaluate secrecy spectral efficiency (SSE), private-rate fairness,
	tag-stream performance, and post-outage reconfiguration.

	The main contributions are as follows:
	\begin{itemize}
		\item We study the system-level integration of PLA and data communication in multi-node networks.
		\item We enable collaborative group-tag transmission through multiple HAPSs, providing distributed delivery and spatial transmission diversity under colluding and non-colluding eavesdropper models.
		\item We jointly optimize HAPS placement, user association, and transmission resources using RCEM as an evolutionary solver and GA2C as a learning-based solver.
	\end{itemize}
	
	Sections~\ref{sec:system_model}--\ref{sec:complexity} present the system model, problem formulation, solution algorithms, and complexity analysis, respectively. Section~\ref{sec:results} reports the numerical results, and Section~\ref{sec:conclusion} concludes the paper.
	
	\subsection*{Mathematical Notation}
	
	Bold lower- and upper-case letters denote vectors and matrices, respectively. $(\cdot)^{\mathsf{T}}$ and $(\cdot)^{\mathsf{H}}$ denote transpose and Hermitian transpose; $(\cdot)^{-1}$ and $(\cdot)^{\dagger}$ denote inverse and Moore--Penrose pseudoinverse. $\|\cdot\|_2$ is the Euclidean norm, and $|\cdot|$ absolute value. $\mathbf{I}_n$ is the $n\times n$ identity matrix. $\mathbb{E}[\cdot]$, $\mathbb{I}\{\cdot\}$, and $\odot$ are expectation, indicator, and Hadamard product.
	$\operatorname{tr}(\cdot)$ denotes the matrix trace, and
	$j\triangleq\sqrt{-1}$ denotes the imaginary unit.
	$\mathcal{CN}(\boldsymbol{\mu},\mathbf{\Sigma})$ denotes a circularly symmetric complex Gaussian distribution with mean $\boldsymbol{\mu}$ and covariance $\mathbf{\Sigma}$. $\mathbb{R}$, $\mathbb{C}$, and $\mathbb{R}_+$ denote reals, complex numbers, and non-negative reals. $[x]^+ \triangleq \max(0,x)$ ensures non-negativity. $\triangleq$ means definition, and $\log_2(\cdot)$ is base-$2$ logarithm. For matrix $\mathbf{A}$, $[\mathbf{A}]_{i,j}$ is its $(i,j)$-th entry, with $\mathbf{A}(i,:)$ and $\mathbf{A}(:,j)$ as its $i$-th row and $j$-th column. $\mathcal{O}(\cdot)$ denotes asymptotic upper bound, and $\Theta(\cdot)$ a tight bound. $|\mathcal{A}|$ is the cardinality of set $\mathcal{A}$.
	
	\section{System Model}
	\label{sec:system_model}
	
	As illustrated in Fig.~\ref{fig:system-model}, we consider an
	NTN comprising a set $\mathcal{K}$ of HAPSs, a set
	$\mathcal{U}$ of single-antenna ground users, and a set
	$\mathcal{E}$ of passive eavesdroppers. The corresponding
	index sets are $\mathcal{K}=\{k_1,k_2,\ldots,k_{|\mathcal{K}|}\}$,
	$\mathcal{U}=\{u_1,u_2,\ldots,u_{|\mathcal{U}|}\}$, and
	$\mathcal{E}=\{e_1,e_2,\ldots,e_{|\mathcal{E}|}\}$.
	
	Let $a_{u,k}\in\{0,1\}$ denote the association indicator, where
	$a_{u,k}=1$ if user $u$ is associated with HAPS $k$, and
	$a_{u,k}=0$ otherwise. The set of users associated with HAPS
	$k$ is defined as
	$\mathcal{U}_k\triangleq
	\{u\in\mathcal{U}\mid a_{u,k}=1\}$, whereas the set of HAPSs
	serving user $u$ is defined as
	$\mathcal{K}_u\triangleq
	\{k\in\mathcal{K}\mid a_{u,k}=1\}$.
	Since multi-connect association is allowed, a user may be
	associated with multiple HAPSs and, consequently, the sets
	$\{\mathcal{U}_k\}_{k\in\mathcal{K}}$ are not necessarily
	disjoint. The set of served users is denoted by
	$\mathcal{U}^{\mathrm{srv}}\triangleq
	\{u\in\mathcal{U}\mid\mathcal{K}_u\neq\emptyset\}$.
	
	Each HAPS $k$ is equipped with $M$ transmit antennas and
	transmits a superimposed signal comprising private,
	group-common, and AN components. Let
	$\mathbf{h}_{k,u}\in\mathbb{C}^{M\times1}$ denote the channel
	vector from HAPS $k$ to user $u$. For a fixed ordering
	$\mathcal{U}_k=\{u_{k,1},\ldots,u_{k,|\mathcal{U}_k|}\}$, the
	aggregated channel matrix from HAPS $k$ to its associated users
	is
	\begin{equation}
		\mathbf{H}_k
		\triangleq
		\begin{bmatrix}
			\mathbf{h}_{k,u_{k,1}}^{\mathrm H}\\
			\vdots\\
			\mathbf{h}_{k,u_{k,|\mathcal{U}_k|}}^{\mathrm H}
		\end{bmatrix}
		\in\mathbb{C}^{|\mathcal{U}_k|\times M}.
	\end{equation}
	
	The user set $\mathcal{U}$ is partitioned into the disjoint
	group-user sets $\{\mathcal{U}_g\}_{g\in\mathcal{G}}$, where
	$\mathcal{G}\triangleq
	\{g_1,\ldots,g_{|\mathcal{G}|}\}$ is the group index set,
	$\mathcal{U}_g\subseteq\mathcal{U}$ is the set of users
	belonging to group $g$, and $N_g\triangleq|\mathcal{U}_g|$ is
	the number of users in that group. Hence,
	$\mathcal{U}=\mathop{\dot{\bigcup}}_{g\in\mathcal{G}}
	\mathcal{U}_g$, where the dotted union indicates that the
	group-user sets are mutually disjoint. The unique group
	containing user $u$ is denoted by $g(u)$.
	
	For each HAPS--group pair $(k,g)$, the associated group-user
	set is defined as
	$\mathcal{U}_{k,g}\triangleq\mathcal{U}_k\cap\mathcal{U}_g$.
	The set of groups served by HAPS $k$ is denoted by
	$\mathcal{G}_k\triangleq
	\{g\in\mathcal{G}\mid\mathcal{U}_{k,g}\neq\emptyset\}$,
	whereas the set of HAPSs serving at least one user in group
	$g$ is denoted by
	$\mathcal{K}_g\triangleq
	\{k\in\mathcal{K}\mid\mathcal{U}_{k,g}\neq\emptyset\}$.
	$\mathcal{K}_g$ is the set of HAPSs transmitting
	the group-common stream.
	Finally, the set of served users belonging to group $g$ is
	defined as
	$\mathcal{U}_g^{\mathrm{srv}}\triangleq
	\mathcal{U}_g\cap\mathcal{U}^{\mathrm{srv}}$.
	
	The groups represent distinct service classes and are used for
	group-common transmission and authentication-tag assignment. The locations
	of the passive eavesdroppers are assumed to be known or
	estimated by the network for secrecy-performance evaluation.
	
	\begin{figure}[t]
		\centering
		\includegraphics[width=\columnwidth]{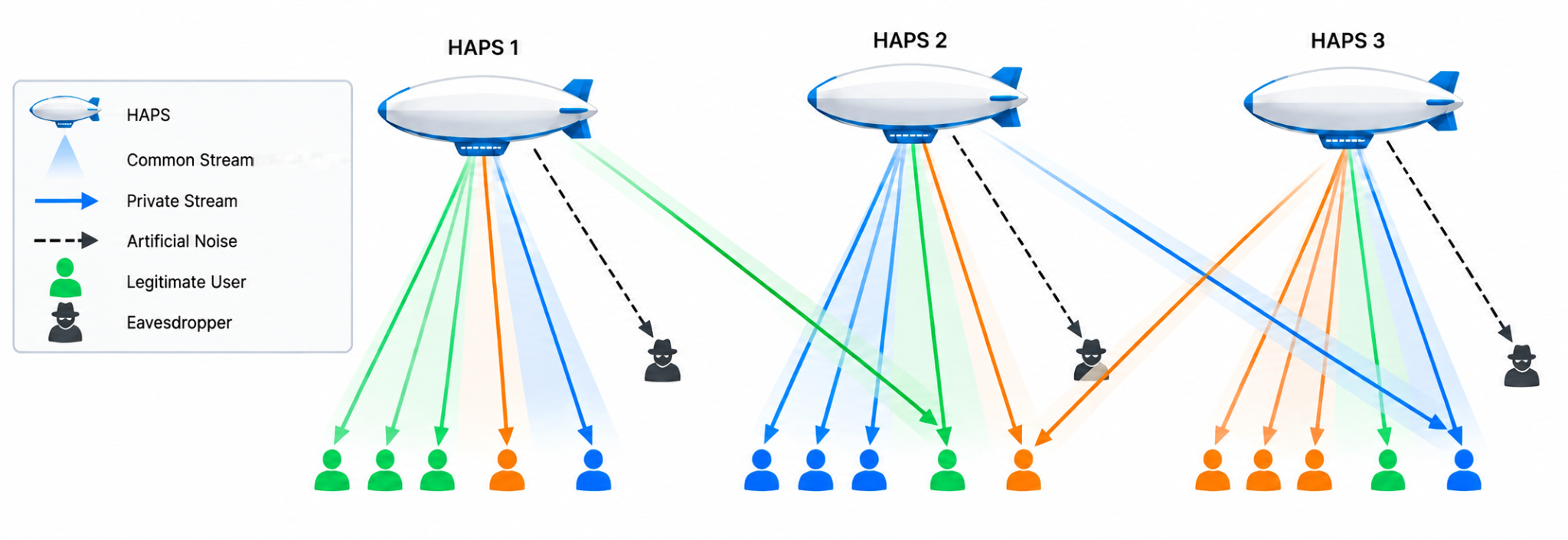}
		\caption{System model.}
		\label{fig:system-model}
	\end{figure}

	\subsection{Authentication Roles and Tag Registration}
	\label{sec:authentication_model}
	
	Each HAPS transmits data and tag streams for its associated users, while allocating power to independent AN.
	Authentication decisions are made locally by users.
	Thus, the authorized HAPSs are authentication claimants, and receiving users are verifiers. The detector checks a registered group tag without establishing individual HAPS identities. Let $\mathcal{C}_{\mathrm{mod}}$ be the unit-energy modulation constellation, with $|c|^2=1$ for each $c\in\mathcal{C}_{\mathrm{mod}}$. For group $g$, $c_g$ is the transmitted symbol and $c_g^{\mathrm{reg}}\in\mathcal{C}_{\mathrm{mod}}$ its registered reference. Non-unit-energy constellations require consistent symbol-energy factors in transmit power, signal-to-interference-plus-noise ratio (SINR), and rate expressions.
	
	Prior studies have integrated authentication with secrecy,
	encoding, or resource allocation at the signal or link level.
	However, they do not address the collaborative delivery of a
	shared group tag and private data by multiple HAPSs under
	joint deployment, user-association, and power-management
	decisions \cite{ref50, flores2026secure}.
	Trusted registration, key establishment, registration-channel protection, and authentication metrics have been studied in~\cite{ref23, zhang2025authscatter, wang2025joint}. Here, the network manager securely provisions a group master key $K_g^{\mathrm{key}}$ to authorized HAPSs and group users. They derive $c_g^{\mathrm{reg}}$ from the key and group identifier using a keyed pseudorandom function and map the resulting bits to the modulation constellation. Provisioning uses an authenticated secure channel and is outside scope; compromised group members and disclosure of $K_g^{\mathrm{key}}$ are excluded. Tags contain multiple symbols, each used for authentication; the formulation considers one representative symbol.
	
	For a user $u\in\mathcal{U}_g$, we consider two
	authentication hypotheses. Under $\mathcal{H}_1$, the
	transmitted tag is the registered tag. Under $\mathcal{H}_0$,
	the transmitted tag is intentionally mismatched:
	\begin{equation}
		\begin{cases}
			\mathcal{H}_1:\quad c_g=c_g^{\mathrm{reg}},\\[4pt]
			\mathcal{H}_0:\quad c_g=c_g^{\mathrm{atk}},
			\qquad
			c_g^{\mathrm{atk}}\in
			\mathcal{C}_{\mathrm{mod}}
			\setminus\{c_g^{\mathrm{reg}}\}.
		\end{cases}
	\end{equation}
	Under $\mathcal{H}_0$, $c_g^{\mathrm{atk}}$ is assumed
	to be selected uniformly from
	$\mathcal{C}_{\mathrm{mod}}
	\setminus\{c_g^{\mathrm{reg}}\}$.
	Let $r_u$ denote the decision statistic obtained by
	user $u$ after common-stream processing.
	The detected tag symbol is
	\begin{equation}
		\hat{c}_u
		\triangleq
		\arg\min_{c\in\mathcal{C}_{\mathrm{mod}}}
		|r_u-c|^2.
	\end{equation}
	The resulting authentication decision is
	\begin{equation}
		\hat{\mathcal{H}}
		=
		\begin{cases}
			\mathcal{H}_1, & \hat{c}_u=c_g^{\mathrm{reg}},\\
			\mathcal{H}_0, & \hat{c}_u\neq c_g^{\mathrm{reg}}.
		\end{cases}
		\label{eq:auth_decision}
	\end{equation}
	
	The statistic $r_u$ is the equalized complex soft symbol
	estimate. Nearest-neighbor detection needs no additional
	scalar threshold. Demodulator design and quantitative
	detection or false-acceptance claims are outside scope.
	The test establishes reference-symbol consistency;
	without freshness or CR, it guarantees neither replay
	resistance nor complete cryptographic authentication.
	
	\subsection{Precoding, Geometry, and Channel Model}
	\label{sec:precoders}
	
	The regularized zero-forcing (RZF) precoder is given by
	\begin{equation}
		\tilde{\mathbf{P}}_k 
		= \mathbf{H}_k^\mathrm{H}\!\left(\mathbf{H}_k \mathbf{H}_k^\mathrm{H} 
		+ \lambda_k \mathbf{I}_{|\mathcal{U}_k|}\right)^{-1},
		\quad \tilde{\mathbf{P}}_k \in \mathbb{C}^{M \times |\mathcal{U}_k|},
	\end{equation}
	where $\lambda_k>0$ is the RZF regularization parameter and 
	$\mathbf{I}_{|\mathcal{U}_k|}$ is the $|\mathcal{U}_k|\times|\mathcal{U}_k|$ identity matrix.
	For $u=u_{k,j}\in\mathcal{U}_k$, normalize the corresponding column of $\tilde{\mathbf{P}}_k$ as
	\begin{equation}
		\mathbf{p}_{k,u}
		=
		\frac{\tilde{\mathbf{P}}_k(:,j)}
		{\|\tilde{\mathbf{P}}_k(:,j)\|_2}.
	\end{equation}
	
	The common beamformer at HAPS $k$ for group $g$ is designed via maximum ratio transmission (MRT) as
	\begin{equation}
		\mathbf{w}_{\mathrm{c},k,g} = 
		\frac{\sum_{u \in \mathcal{U}_{k,g}} \alpha_{k,u}\,\mathbf{h}_{k,u}}
		{\left\|\sum_{u \in \mathcal{U}_{k,g}} \alpha_{k,u}\,\mathbf{h}_{k,u}\right\|},
	\end{equation}
	where $\alpha_{k,u} \geq 0$ are weighting coefficients.
	
	To enhance confidentiality, HAPS $k$ injects AN into the null space of the channel matrix $\mathbf{H}_k$.
	The corresponding orthogonal projection matrix is
	\begin{equation}
		\mathbf{P}_{\perp,k}
		=
		\mathbf{I}_M-\mathbf{H}_k^\dagger\mathbf{H}_k,
		\label{eq:null_projection}
	\end{equation}
	where $\mathbf{H}_k^\dagger$ denotes the Moore--Penrose
	pseudoinverse.
	The matrix
	$\mathbf{P}_{\perp,k}\in\mathbb{C}^{M\times M}$
	is Hermitian and idempotent and satisfies $\mathbf{H}_k\mathbf{P}_{\perp,k}=\mathbf{0}$.
	If $\mathbf{H}_k$ has full row rank, an equivalent
	expression is
	\begin{equation}
		\mathbf{P}_{\perp,k}
		=
		\mathbf{I}_M
		-
		\mathbf{H}_k^{\mathrm H}
		\left(
		\mathbf{H}_k\mathbf{H}_k^{\mathrm H}
		\right)^{-1}
		\mathbf{H}_k.
	\end{equation}
	
	Each HAPS $k$ is located at
	$
	\mathbf{q}_k
	=
	[q_{k1},q_{k2},q_{k3}]^{\mathsf T}
	\in\mathbb{R}^{3\times1}
	$,
	where $q_{k1}$ and $q_{k2}$ are its horizontal coordinates
	and $q_{k3}>0$ is its altitude.
	User $u$ and eavesdropper
	$e$ are located at
	$
	\mathbf{q}_u
	=
	[q_{u1},q_{u2},0]^{\mathsf T}$, and $
	\mathbf{q}_e
	=
	[q_{e1},q_{e2},0]^{\mathsf T}
	$, respectively.
	
	The three-dimensional distance and elevation angle
	between HAPS $k$ and user $u$ are
	\begin{equation}
		d_{k,u}
		=
		\|\mathbf{q}_k-\mathbf{q}_u\|_2,
		\qquad
		\theta_{k,u}
		=
		\arcsin\left(\frac{q_{k3}}{d_{k,u}}\right),
	\end{equation}
	where
	$d_{k,u}>0$ and
	$\theta_{k,u}\in[0,\pi/2]$.
	
	We adopt a probabilistic LoS model:
	\begin{equation}
		P_{\mathrm{LoS}}(\theta_{k,u})
		=
		\frac{1}{
			1+a_l
			\exp\left[-b_l(\theta_{k,u}-\theta_0)\right]
		},
		\label{eq:plos}
	\end{equation}
	where $a_l>0$ and $b_l>0$ are environment-dependent
	parameters and $\theta_0$ is a reference elevation angle.
	The quantities $\theta_{k,u}$, $\theta_0$, and the unit
	associated with $b_l$ are represented consistently.
	For each HAPS--user link, the LoS state is drawn from a
	Bernoulli distribution with parameter
	$P_{\mathrm{LoS}}(\theta_{k,u})$.
	
	Conditioned on the propagation state
	$s\in\{\mathrm{LoS},\mathrm{NLoS}\}$, the large-scale
	channel power gain is
	\begin{equation}
		\beta_{k,u}
		=
		\frac{
			C_0d_{k,u}^{-\alpha_s}
			10^{-S_{k,u}/10}
		}{
			L_{\mathrm{atm}}(d_{k,u})
		},
		\qquad
		\beta_{k,u}>0,
		\label{eq:beta}
	\end{equation}
	where $C_0>0$ is a frequency-dependent reference gain,
	$\alpha_s>0$ is the path-loss exponent for state $s$,
	$S_{k,u}\sim\mathcal{N}(0,\sigma_S^2)$ is the
	shadowing variable in dB, where $\sigma_S>0$
	is its standard deviation, and
	$L_{\mathrm{atm}}(d_{k,u})\geq1$ is the atmospheric
	power-loss factor.
	Thus, the channel vector is modeled as
	\begin{equation}
		\footnotesize
		\mathbf{h}_{k,u}
		=
		\begin{cases}
			\sqrt{\beta_{k,u}}
			\left(
			\sqrt{\dfrac{\kappa_r}{\kappa_r+1}}\,
			\mathbf{h}_{k,u}^{\mathrm{LoS}}
			+
			\sqrt{\dfrac{1}{\kappa_r+1}}\,
			\mathbf{h}_{k,u}^{\mathrm{NLoS}}
			\right),
			& \mathrm{LoS},\\[6pt]
			\sqrt{\beta_{k,u}}\,
			\mathbf{h}_{k,u}^{\mathrm{NLoS}},
			& \mathrm{NLoS},
		\end{cases}
		\normalsize
		\label{eq:rician}
	\end{equation}
	where
	$\mathbf{h}_{k,u}^{\mathrm{NLoS}}
	\sim\mathcal{CN}(\mathbf{0},\mathbf{I}_M)$
	is the scattering component,
	$\mathbf{h}_{k,u}^{\mathrm{LoS}}
	\in\mathbb{C}^{M\times1}$
	is the deterministic array steering vector, and
	$\kappa_r\geq0$ is the Rician factor.
	In the simulation, the HAPS--eavesdropper channels are generated using the
	same propagation model after replacing the user location
	$\mathbf{q}_u$ with the eavesdropper location
	$\mathbf{q}_e$.

	\subsection{Transmit Model and Association}
	
	In the exclusive-association mode,
	$|\mathcal{K}_u|\leq1$, whereas in the multi-connect mode a
	served user may be associated with multiple HAPSs. A user
	with $\mathcal{K}_u=\emptyset$ is considered unserved.
	
	Let $U_{\mathrm{perRB}}\in\mathbb{Z}_+$ denote the maximum
	aggregate number of users associated with one HAPS under the
	considered resource block (RB) configuration. The per-HAPS association
	constraint is
	\begin{equation}
		\sum_{u\in\mathcal{U}}a_{u,k}
		\leq U_{\mathrm{perRB}},
		\qquad k\in\mathcal{K}.
	\end{equation}
	Because every user belongs to exactly one group, this
	aggregate constraint also upper-bounds the number of users
	served by a HAPS on each group-specific RB.
	
	The transmission model is multi-layer rate splitting multiple access (RSMA)-inspired \cite{mao2022rate}. The group-common stream carries authentication tag $c_g$, while the private stream carries confidential data symbol $b_u$. Accordingly, $R_{\mathrm c,g}$ is the tag-stream rate proxy and $R_{\mathrm p,u}$ is the achievable private-data rate; their secrecy metrics are evaluated separately.
	
	Groups occupy mutually orthogonal frequency RBs
	that are active simultaneously.
	On the RB assigned to group $g$, HAPS $k\in\mathcal{K}_g$
	transmits
	\begin{equation}
		\begin{aligned}
			&\mathbf{x}_{k,g}=\\
			&\sqrt{P_{\mathrm{t},k}}\Big(
			\rho_{\mathrm{b},k}
			\sum_{u\in\mathcal{U}_{k,g}}\mathbf{p}_{k,u}b_u
			+\rho_{\mathrm{c},k,g}
			\mathbf{w}_{\mathrm{c},k,g}c_g
			+\rho_{\mathrm{d},k}\mathbf{v}_{k,g}
			\Big),
		\end{aligned}
		\label{eq:signal}
	\end{equation}
	and $\mathbf{x}_{k,g}=\mathbf{0}$ otherwise.
	Here, $P_{\mathrm{t},k}$ is the total power budget
	across the active RBs of HAPS $k$.
	The private data symbol of user $u$ is denoted by $b_u$.
	The vectors $\mathbf{p}_{k,u}$ and
	$\mathbf{w}_{\mathrm{c},k,g}$ have unit norm.
	The private data symbols are mutually independent and satisfy
	$\mathbb{E}[b_u]=0$ and $\mathbb{E}[|b_u|^2]=1$. The tag
	symbol satisfies $|c_g|^2=1$ and is independent of the private
	data symbols and AN vectors.
	The nonnegative coefficients
	$\rho_{\mathrm{b},k}$, $\rho_{\mathrm{c},k,g}$,
	and $\rho_{\mathrm{d},k}$ allocate transmit amplitudes.
	Each active HAPS--group transmission jointly carries the
	group authentication tag and the private data of the
	associated users in that group. Thus, every HAPS serving
	user $u$ transmits both its private stream and its group
	tag on the same RB; no standalone tag-only transmission
	is introduced. User $u$ may additionally receive copies
	of the same tag from HAPSs serving other users in its group.
	For each active group resource block, let
	$\mathbf{z}_{k,g}\sim\mathcal{CN}(\mathbf{0},\mathbf{I}_M)$
	be an independent Gaussian vector. Define
	\begin{equation}
		\nu_k
		\triangleq
		\operatorname{tr}\!\left(
		\mathbf{P}_{\perp,k}\mathbf{P}_{\perp,k}^{\mathrm H}
		\right).
	\end{equation}
	For $\nu_k>0$, the unit-average-power AN vector is
	\begin{equation}
		\mathbf{v}_{k,g}
		\triangleq
		\begin{cases}
			\dfrac{\mathbf{P}_{\perp,k}\mathbf{z}_{k,g}}
			{\sqrt{\nu_k}}, & \nu_k>0,\\[8pt]
			\mathbf{0}, & \nu_k=0.
		\end{cases}
	\end{equation}
	For $\nu_k>0$,
	$\mathbb{E}[\|\mathbf{v}_{k,g}\|_2^2]=1$; otherwise, AN
	transmission is disabled by setting $\rho_{\mathrm{d},k}=0$.
	For the normalized AN, define the spatial covariance matrix as
	\begin{equation}
		\mathbf{Q}_{k}^{\mathrm{AN}}
		\triangleq
		\mathbb{E}\!\left[
		\mathbf{v}_{k,g}\mathbf{v}_{k,g}^{\mathrm H}
		\right]
		=
		\begin{cases}
			\dfrac{
				\mathbf{P}_{\perp,k}\mathbf{P}_{\perp,k}^{\mathrm H}
			}{
				\nu_k
			}, & \nu_k>0,\\[8pt]
			\mathbf{0}, & \nu_k=0.
		\end{cases}
		\label{eq:AN_covariance}
	\end{equation}
	AN realizations are independent across HAPSs and RBs and independent of the information streams. Their covariance is identical across group RBs and hence independent of $g$.
	
	The total average transmit power of HAPS $k$ is
	\begin{equation}
		\begin{aligned}
			\sum_{g\in\mathcal{G}_k}
			\mathbb{E}\|\mathbf{x}_{k,g}\|_2^2
			=
			P_{\mathrm{t},k}\Big(
			|\mathcal{U}_k|\rho_{\mathrm{b},k}^2
			+\sum_{g\in\mathcal{G}_k}\rho_{\mathrm{c},k,g}^2
			+|\mathcal{G}_k|\rho_{\mathrm{d},k}^2
			\Big).
		\end{aligned}
		\label{eq:actual_power}
	\end{equation}
	
	Hence, satisfying the total power budget
	$P_{\mathrm{t},k}$ requires
	\begin{equation}
		|\mathcal{U}_k|\rho_{\mathrm{b},k}^2
		+
		\sum_{g\in\mathcal{G}_k}
		\rho_{\mathrm{c},k,g}^2
		+
		|\mathcal{G}_k|\rho_{\mathrm{d},k}^2
		\leq 1.
		\label{eq:power_feasibility}
	\end{equation}
	
	HAPS transmissions are time/frequency aligned under
	per-RB flat fading. Effective channels include residual
	carrier phases, and stream copies add as complex amplitudes.
	
	The aggregate transmitted waveform of HAPS $k$ is
	\begin{equation}
		\begin{aligned}
			\mathbf{x}_k(\tau)
			=\sqrt{P_{\mathrm{t},k}}\Bigg(
			&\rho_{\mathrm{b},k}
			\sum_{u\in\mathcal{U}_k}
			\mathbf{p}_{k,u}b_u\,\varphi_{g(u)}(\tau)\\
			&+\sum_{g\in\mathcal{G}_k}
			\rho_{\mathrm{c},k,g}
			\mathbf{w}_{\mathrm{c},k,g}c_g\,\varphi_g(\tau)
			+\rho_{\mathrm{d},k}\mathbf{v}_k(\tau)
			\Bigg),
			\label{eq:aggregate_transmit_waveform}
		\end{aligned}
	\end{equation}
	where $\tau$ denotes continuous time within a signaling
	interval of duration $T$, and $\varphi_g(\tau)$ is the
	normalized waveform assigned to the RB of group $g$,
	satisfying
	\begin{equation}
		\frac{1}{T}\int_0^T
		\varphi_g(\tau)\varphi_{g'}^*(\tau)\,\mathrm{d}\tau
		=
		\begin{cases}
			1, & g=g',\\
			0, & g\neq g'.
		\end{cases}
	\end{equation}
	The aggregate AN waveform is
	\begin{equation}
		\mathbf{v}_k(\tau)
		\triangleq
		\sum_{g\in\mathcal{G}_k}
		\mathbf{v}_{k,g}\varphi_g(\tau).
	\end{equation}
	Thus, equivalently,
	$\mathbf{x}_k(\tau)=
	\sum_{g\in\mathcal{G}_k}\mathbf{x}_{k,g}\varphi_g(\tau)$.

	For $u\in\mathcal{U}_g$, the received signal is
	\begin{equation}
		y_{u,g}
		=
		\sum_{k\in\mathcal{K}_g}
		\mathbf{h}_{k,u}^{\mathrm H}\mathbf{x}_{k,g}
		+n_{u,g},
	\end{equation}
	where $n_{u,g}\sim\mathcal{CN}(0,\sigma^2)$
	is receiver noise on the considered RB.
	Therefore, the common-stream SINR is
	\begin{equation}
		\gamma_{\mathrm{c},u}
		=
		\frac{
			\left|
			\sum_{k\in\mathcal{K}_g}
			\sqrt{P_{\mathrm{t},k}}\rho_{\mathrm{c},k,g}
			\mathbf{h}_{k,u}^{\mathrm H}\mathbf{w}_{\mathrm{c},k,g}
			\right|^2
		}{
			I_{u,g}^{\mathrm{pri}}+I_{u,g}^{\mathrm{AN}}+\sigma^2
		},
		\label{eq:gamma_common}
	\end{equation}
	where the interference from the private streams of users in group $g$ is given by
	\begin{equation}
		I_{u,g}^{\mathrm{pri}}
		=
		\sum_{u'\in\mathcal{U}_g}
		\left|
		\sum_{k\in\mathcal{K}_{u'}}
		\sqrt{P_{\mathrm{t},k}}\rho_{\mathrm{b},k}
		\mathbf{h}_{k,u}^{\mathrm H}\mathbf{p}_{k,u'}
		\right|^2
		\label{eq:I_common_private}
	\end{equation}
	and the interference caused by AN is expressed as
	\begin{equation}
		I_{u,g}^{\mathrm{AN}}
		=
		\sum_{k\in\mathcal{K}_g}
		P_{\mathrm{t},k}\rho_{\mathrm{d},k}^2
		\mathbf{h}_{k,u}^{\mathrm H}
		\mathbf{Q}_{k}^{\mathrm{AN}}
		\mathbf{h}_{k,u}.
		\label{eq:I_user_AN}
	\end{equation}
	Empty sums are interpreted as zero.
	The AN of HAPS $k$ is nulled at its associated users
	under perfect CSI, but may interfere with other users
	on the same RB.
	Accordingly, aggregate AN interference is retained.
	
	The achievable common rate is
	\begin{equation}
		R_{c,g}
		=
		\min_{u\in\mathcal{U}_g^{\mathrm{srv}}}
		\log_2(1+\gamma_{\mathrm{c},u}).
	\end{equation}
	We set $R_{c,g}=0$ for an inactive group.
	
	Each served user first decodes its group-common
	stream while treating private streams as interference.
	After successful decoding, perfect successive interference cancellation (SIC) removes
	the common stream.
	The private-stream SINR is then
	\begin{equation}
		\gamma_{\mathrm{p},u}
		=
		\frac{
			\left|
			\sum_{k\in\mathcal{K}_u}
			\sqrt{P_{\mathrm{t},k}}\rho_{\mathrm{b},k}
			\mathbf{h}_{k,u}^{\mathrm H}\mathbf{p}_{k,u}
			\right|^2
		}{
			I_u^{\mathrm{MUI}}+I_{u,g(u)}^{\mathrm{AN}}+\sigma^2
		},
	\end{equation}
	where the multi-user interference (MUI) is defined as
	\begin{equation}
		I_u^{\mathrm{MUI}}
		=
		\sum_{\substack{u'\in\mathcal{U}_{g(u)}\\u'\neq u}}
		\left|
		\sum_{k\in\mathcal{K}_{u'}}
		\sqrt{P_{\mathrm{t},k}}\rho_{\mathrm{b},k}
		\mathbf{h}_{k,u}^{\mathrm H}\mathbf{p}_{k,u'}
		\right|^2.
		\label{eq:I_private_MUI}
	\end{equation}
	Thus, the achievable private rate for user $u$ is
	\begin{equation}
		R_{\mathrm{p},u}=\log_2(1+\gamma_{\mathrm{p},u}),
	\end{equation}
	with $R_{\mathrm{p},u}=0$ for unserved users.
	These rates are normalized by the bandwidth of
	the corresponding RB.
	
	The channel vector from HAPS $k$ to eavesdropper
	$e\in\mathcal{E}$ is denoted by
	$\mathbf{h}_{k,e}\in\mathbb{C}^{M\times1}$.
	The receiver noise at each eavesdropper is assumed
	independent and distributed as
	$\mathcal{CN}(0,\sigma^2)$.

	Eavesdroppers know the signaling structure and estimate their effective channels, but do not know the group key. Colluding eavesdroppers share complex baseband observations and channel estimates for joint linear combining.
	
	When eavesdropper $e$ attempts to decode the private
	stream of user $u$, its SINR is
	\begin{equation}
		\gamma_{e,u}
		=
		\frac{
			\left|
			\displaystyle
			\sum_{k\in\mathcal{K}_u}
			\sqrt{P_{\mathrm{t},k}}\rho_{\mathrm{b},k}
			\mathbf{h}_{k,e}^{\mathrm H}\mathbf{p}_{k,u}
			\right|^2
		}{
			I_{e,u}^{\mathrm{MUI}}
			+I_{e,g(u)}^{\mathrm C}
			+I_{e,g(u)}^{\mathrm{AN}}
			+\sigma^2
		}.
		\label{eq:eavesdropper_SINR}
	\end{equation}
	Here, $\gamma_{e,u}$ is the scalar SINR observed by
	eavesdropper $e$ for the private stream of user $u$.
	
	The multi-user interference observed by eavesdropper $e$
	while decoding user $u$ is
	\begin{equation}
		I_{e,u}^{\mathrm{MUI}}
		\triangleq
		\sum_{\substack{
				u'\in\mathcal{U}_{g(u)}\\
				u'\neq u
		}}
		\left|
		\sum_{k\in\mathcal{K}_{u'}}
		\sqrt{P_{\mathrm{t},k}}\rho_{\mathrm{b},k}
		\mathbf{h}_{k,e}^{\mathrm H}\mathbf{p}_{k,u'}
		\right|^2.
		\label{eq:I_e_MUI}
	\end{equation}
	The inner sum combines copies of one private symbol;
	the outer sum adds independent-stream powers.
	
	The interference caused by the common stream of group
	$g$ at eavesdropper $e$ is
	\begin{equation}
		I_{e,g}^{\mathrm C}
		\triangleq
		\left|
		\sum_{k\in\mathcal{K}_g}
		\sqrt{P_{\mathrm{t},k}}\rho_{\mathrm{c},k,g}
		\mathbf{h}_{k,e}^{\mathrm H}
		\mathbf{w}_{\mathrm{c},k,g}
		\right|^2.
		\label{eq:I_e_common}
	\end{equation}
	
	The AN power observed by eavesdropper $e$ on the RB of
	group $g$ is
	\begin{equation}
		I_{e,g}^{\mathrm{AN}}
		\triangleq
		\sum_{k\in\mathcal{K}_g}
		P_{\mathrm{t},k}\rho_{\mathrm{d},k}^2
		\mathbf{h}_{k,e}^{\mathrm H}
		\mathbf{Q}_{k}^{\mathrm{AN}}
		\mathbf{h}_{k,e}.
		\label{eq:I_e_AN}
	\end{equation}
	
	Eavesdroppers treat other streams as interference without
	SIC. Secrecy rates apply to this receiver model and exclude
	common-stream decoding and cancellation.
	
	\paragraph{Non-Colluding Eavesdroppers.}
	For independently operating eavesdroppers, define
	\begin{equation}
		R_{e,u}
		\triangleq
		\log_2(1+\gamma_{e,u})
	\end{equation}
	as the rate at which eavesdropper $e$ can decode the
	private stream of user $u$ under the stated decoding
	model. The corresponding secrecy-rate metric is
	\begin{equation}
		R_{\mathrm{s},u}^{\mathrm{nc}}
		=
		\left[
		R_{\mathrm{p},u}
		-
		\max_{e\in\mathcal{E}}R_{e,u}
		\right]^+,
		\label{eq:secrecy_noncolluding}
	\end{equation}
	where the superscript $\mathrm{nc}$ denotes the
	non-colluding eavesdropper model.
	
	\paragraph{Colluding Eavesdroppers.}
	In the colluding model, all eavesdroppers share their
	complex baseband observations and jointly decode the
	private stream of user $u$.
	Let
	$\mathbf{f}_u\in\mathbb{C}^{|\mathcal{E}|\times1}$
	denote the effective channel vector of that private
	stream across all eavesdroppers. Its entry associated
	with eavesdropper $e$ is
	\begin{equation}
		[\mathbf{f}_u]_e
		\triangleq
		\sum_{k\in\mathcal{K}_u}
		\sqrt{P_{\mathrm{t},k}}\rho_{\mathrm{b},k}
		\mathbf{h}_{k,e}^{\mathrm H}\mathbf{p}_{k,u}.
		\label{eq:colluding_desired_vector}
	\end{equation}
	The notation $[\mathbf{f}_u]_e$ denotes the entry of
	$\mathbf{f}_u$ corresponding to eavesdropper $e$.
	
	For every interfering private stream $u'$, the vector
	$\mathbf{f}_{u'}\in\mathbb{C}^{|\mathcal{E}|\times1}$
	is defined in the same manner as
	\eqref{eq:colluding_desired_vector}, after replacing
	user $u$ by user $u'$.
	The effective common-stream vector of group $g$ is denoted
	by $\mathbf{d}_g\in\mathbb{C}^{|\mathcal{E}|\times1}$,
	whose entry associated with eavesdropper $e$ is
	\begin{equation}
		[\mathbf{d}_g]_e
		\triangleq
		\sum_{k\in\mathcal{K}_g}
		\sqrt{P_{\mathrm{t},k}}\rho_{\mathrm{c},k,g}
		\mathbf{h}_{k,e}^{\mathrm H}
		\mathbf{w}_{\mathrm{c},k,g}.
		\label{eq:colluding_common_vector}
	\end{equation}
	
	Let
	$\mathbf{C}_g^{\mathrm{AN}}
	\in\mathbb{C}^{|\mathcal{E}|\times|\mathcal{E}|}$
	denote the covariance matrix of the AN observed jointly
	by the eavesdroppers on the RB of group $g$.
	For $e,e'\in\mathcal{E}$, its $(e,e')$-th entry is
	\begin{equation}
		\left[
		\mathbf{C}_g^{\mathrm{AN}}
		\right]_{e,e'}
		\triangleq
		\sum_{k\in\mathcal{K}_g}
		P_{\mathrm{t},k}\rho_{\mathrm{d},k}^2
		\mathbf{h}_{k,e}^{\mathrm H}
		\mathbf{Q}_{k}^{\mathrm{AN}}
		\mathbf{h}_{k,e'}.
		\label{eq:colluding_AN_covariance}
	\end{equation}
	Here, $e'$ is the index of a second eavesdropper, and the
	off-diagonal entries describe the correlation of the AN
	observations between different eavesdroppers.
	
	The interference-plus-noise covariance matrix for jointly
	decoding the private stream of user $u$ is
	\begin{equation}
		\mathbf{C}_u
		\triangleq
		\sum_{\substack{
				u'\in\mathcal{U}_{g(u)}\\
				u'\neq u
		}}
		\mathbf{f}_{u'}\mathbf{f}_{u'}^{\mathrm H}
		+
		\mathbf{d}_{g(u)}\mathbf{d}_{g(u)}^{\mathrm H}
		+
		\mathbf{C}_{g(u)}^{\mathrm{AN}}
		+
		\sigma^2\mathbf{I}_{|\mathcal{E}|},
		\label{eq:colluding_interference_covariance}
	\end{equation}
	where
	$\mathbf{C}_u\in
	\mathbb{C}^{|\mathcal{E}|\times|\mathcal{E}|}$
	and $\mathbf{I}_{|\mathcal{E}|}$ is the
	$|\mathcal{E}|\times|\mathcal{E}|$ identity matrix.
	Because $\sigma^2>0$, the matrix $\mathbf{C}_u$ is
	positive definite and invertible.
	
	Under optimal linear combining, the aggregate SINR of
	the colluding eavesdroppers is
	\begin{equation}
		\gamma_{\mathrm{e},u}^{\mathrm c}
		\triangleq
		\mathbf{f}_u^{\mathrm H}
		\mathbf{C}_u^{-1}
		\mathbf{f}_u,
		\label{eq:colluding_SINR}
	\end{equation}
	where the superscript $\mathrm c$ denotes the colluding
	eavesdropper model.
	The corresponding secrecy-rate metric is
	\begin{equation}
		R_{\mathrm{s},u}^{\mathrm c}
		=
		\left[
		R_{\mathrm{p},u}
		-
		\log_2\!
		\left(
		1+\gamma_{\mathrm{e},u}^{\mathrm c}
		\right)
		\right]^+.
		\label{eq:secrecy_colluding}
	\end{equation}
	
	If $\mathbf{C}_u$ is diagonal, the observations of the
	eavesdroppers have uncorrelated effective interference
	and noise. In that special case,
	\begin{equation}
		\gamma_{\mathrm{e},u}^{\mathrm c}
		=
		\sum_{e\in\mathcal{E}}\gamma_{e,u}.
	\end{equation}

	\subsection{SSE Metrics and User-Rate Fairness}
	\label{subsec:sse_metrics}
	
	All SSEs are expressed in bit/s/Hz using the bandwidth
	of the corresponding group RB. Let
	$\mathcal{M}\triangleq\{\mathrm{nc},\mathrm{c}\}$
	denote the non-colluding and colluding eavesdropper models,
	respectively.
	
	\paragraph{Private-data SSE.}
	The user-level SSE is
	\begin{equation}
		\mathrm{SSE}_{u}^{m}
		\triangleq R_{\mathrm{s},u}^{m},
		\label{eq:sse_user}
	\end{equation}
	where $R_{\mathrm{s},u}^{m}$ is given by
	\eqref{eq:secrecy_noncolluding} or
	\eqref{eq:secrecy_colluding} and $m\in\mathcal{M}$.
	Unserved users contribute zero.
	The group-level and association-based HAPS-level SSEs are
	\begin{align}
		\mathrm{SSE}_{g}^{m}
		&\triangleq
		\sum_{u\in\mathcal{U}_g^{\mathrm{srv}}}
		\mathrm{SSE}_{u}^{m},
		\label{eq:sse_group}\\
		\mathrm{SSE}_{k}^{m}
		&\triangleq
		\sum_{u\in\mathcal{U}_k}
		\mathrm{SSE}_{u}^{m}.
		\label{eq:sse_haps}
	\end{align}
	Here, $\mathrm{SSE}_{g}^{m}$ aggregates private-data secrecy
	in group $g$. The HAPS metric counts each associated user's
	full SSE, rather than an independent HAPS stream rate:
	
	\[
	\sum_{k\in\mathcal{K}}\mathrm{SSE}_{k}^{m}
	=
	\sum_{u\in\mathcal{U}^{\mathrm{srv}}}
	|\mathcal{K}_u|\,\mathrm{SSE}_{u}^{m}.
	\]
	For HAPS-level and group-level statistics, sum, mean,
	and minimum are taken over $\mathcal{K}$ and
	$\mathcal{G}$, respectively.
	
	The aggregate private-data SSE and its per-user average are
	\begin{align}
		\mathrm{SSE}_{\mathrm{tot}}^{m}
		&\triangleq
		\sum_{g\in\mathcal{G}}\mathrm{SSE}_{g}^{m}
		=
		\sum_{u\in\mathcal{U}}\mathrm{SSE}_{u}^{m},
		\label{eq:sse_total}\\
		\overline{\mathrm{SSE}}^{m}
		&\triangleq
		\frac{\mathrm{SSE}_{\mathrm{tot}}^{m}}
		{|\mathcal{U}|}.
		\label{eq:sse_average}
	\end{align}
	Each user is counted once, with zero for unserved users.
	No bandwidth weights are applied, so
	$\mathrm{SSE}_{\mathrm{tot}}^{m}$ retains group-RB normalization.
	For non-colluding eavesdroppers, the conditional mean of \eqref{eq:sse_average} is derived in Appendix~\ref{app:sse_proof}.
	
	\paragraph{Authentication-tag SSE.}
	For the tag-carrying common stream of group $g$, define
	\begin{equation}
		\begin{aligned}
			\gamma_{\mathrm{e,c},g}
			&\triangleq
			\frac{I_{e,g}^{\mathrm C}}
			{I_{e,g}^{\mathrm{pri}}
				+I_{e,g}^{\mathrm{AN}}+\sigma^2},\\
			\mathrm{SSE}_{\mathrm{tag},g}^{\mathrm{nc}}
			&\triangleq
			\left[
			R_{\mathrm c,g}
			-\max_{e\in\mathcal E}
			\log_2(1+\gamma_{\mathrm{e,c},g})
			\right]^+.
		\end{aligned}
		\label{eq:tag_secrecy_nc}
	\end{equation}
	Here, $I_{e,g}^{\mathrm{pri}}$ is the total received
	power of all private streams on group $g$'s RB. It is expressed as
	\begin{equation}
	I_{e,g}^{\mathrm{pri}}
	\triangleq
	\sum_{u'\in\mathcal U_g}
	\left|
	\sum_{k\in\mathcal K_{u'}}
	\sqrt{P_{\mathrm t,k}}\rho_{\mathrm b,k}
	\mathbf h_{k,e}^{\mathrm H}\mathbf p_{k,u'}
	\right|^2.
	\end{equation}
	Each eavesdropper independently decodes the common
	stream while treating private streams and AN as
	interference. This common-stream metric is distinct
	from the private-data group metric in
	\eqref{eq:sse_group}.
	
	The group-averaged authentication-tag SSE is
	\begin{equation}
		\overline{\mathrm{SSE}}_{\mathrm{tag}}^{\mathrm{nc}}
		\triangleq
		\frac{1}{|\mathcal{G}|}
		\sum_{g\in\mathcal{G}}
		\mathrm{SSE}_{\mathrm{tag},g}^{\mathrm{nc}}.
		\label{eq:sse_tag_average}
	\end{equation}
	
	\paragraph{Retained SSE under HAPS outages.}
	Let $\mathsf{A}\in\{\mathrm{RCEM},\mathrm{GA2C}\}$
	denote the method and
	$s\in\{1,\ldots,N_{\mathrm{seed}}\}$ the seed.
	Let $\mathcal{F}_j^{(s)}$ contain the tested sets
	of $j$ failed HAPSs under seed $s$, with
	$\mathcal{F}_0^{(s)}=\{\emptyset\}$.
	Denote the private-data secrecy rate after
	reoptimization for failure set $F$ by
	$R_{\mathrm{s},u}^{\mathsf A,s,m}(F)$.
	The per-seed average over failure sets and the
	unchanged original user set is
	\begin{equation}
		\overline{\mathrm{SSE}}^{\mathsf A,s,m}(j)
		=
		\frac{1}{|\mathcal F_j^{(s)}|\,|\mathcal U|}
		\sum_{F\in\mathcal F_j^{(s)}}
		\sum_{u\in\mathcal U}
		R_{\mathrm s,u}^{\mathsf A,s,m}(F).
		\label{eq:outage_seed_sse}
	\end{equation}
	Unserved users again contribute zero.
	For positive pre-outage averages, define
	\begin{equation}
		\mathrm{SSE}_{\mathrm{ret}}^{\mathsf A,s,m}(j)
		=
		100
		\frac{
			\overline{\mathrm{SSE}}^{\mathsf A,s,m}(j)
		}{
			\overline{\mathrm{SSE}}^{\mathsf A,s,m}(0)
		}.
		\label{eq:outage_seed_retention}
	\end{equation}
	The reported retained-SSE percentage is
	\begin{equation}
		\overline{\mathrm{SSE}}_{\mathrm{ret}}^{\mathsf A,m}(j)
		=
		\frac{1}{N_{\mathrm{seed}}}
		\sum_{s=1}^{N_{\mathrm{seed}}}
		\mathrm{SSE}_{\mathrm{ret}}^{\mathsf A,s,m}(j).
		\label{eq:outage_mean_retention}
	\end{equation}
	This averages seed-level percentages rather than taking
	a ratio of seed-averaged SSEs. Each method uses its own
	pre-outage denominator, so higher retention need not
	imply higher absolute SSE.
	
	Jain's index evaluates equality among achievable private-stream rates:
	
	\begin{equation}
		\mathcal{J}(\mathbf{R}) 
		= \frac{\left( \sum_{u \in \mathcal{U}} R_{\mathrm{p},u} \right)^{2}}
		{|\mathcal{U}| \sum_{u \in \mathcal{U}} R_{\mathrm{p},u}^{2}}.
		\label{eq:jain}
	\end{equation}
	Here, $\mathbf{R}\triangleq [\,R_{\mathrm{p},u}\,]_{u\in\mathcal{U}} \in \mathbb{R}_+^{|\mathcal{U}|}$ represents the per-user achievable rates. Therefore, $\mathcal{J}=1$ indicates perfect equality, while $\mathcal{J}\to 1/|\mathcal{U}|$ indicates extreme unfairness.
	Unserved users have zero private rate.
	For a nonzero rate vector,
	$\mathcal{J}(\mathbf{R})\in[1/|\mathcal{U}|,1]$.
	When all rates are zero, we set
	$\mathcal{J}(\mathbf{R})=0$ by convention.
	
	\section{Problem Formulation}
	\label{sec:problem}
	
	We optimize private-rate fairness and deployment/association
	regularity subject to tag-delivery, secrecy, and service
	constraints. SSE is evaluated as an outcome metric rather
	than directly maximized.
	
	\paragraph{Graph-Regularized HAPS Layout}
	Stack the HAPS coordinates as \(\mathbf{Q}\! =\![\mathbf{q}_1^\mathsf{T};\ldots;\mathbf{q}_{|\mathcal{K}|}^\mathsf{T}]\in\mathbb{R}^{|\mathcal{K}|\times 3}\).
	Let $\mathbf W\in\mathbb R^{|\mathcal K|\times|\mathcal K|}$
	be a symmetric edge-weight matrix with zero diagonal.
	For distinct HAPSs $i$ and $j$, define the rate-weighted
	shared-user coupling metric as
	\begin{equation}
		R^{\mathrm{hh}}_{i,j}
		\triangleq
		\sum_{u\in\mathcal U_i\cap\mathcal U_j}R_{\mathrm p,u}
		=
		\sum_{u\in\mathcal U}a_{u,i}a_{u,j}R_{\mathrm p,u}.
		\label{eq:shared_user_coupling}
	\end{equation}
	Thus, $R^{\mathrm{hh}}_{i,j}=R^{\mathrm{hh}}_{j,i}\geq0$.
	The metric aggregates the private-stream rates of users
	jointly served by the two HAPSs and represents their
	shared-user coupling. For fixed user rates, adding a shared
	user with a positive rate increases this metric.
	The edge weights are
	\begin{equation}
		[\mathbf W]_{i,j}
		=
		\exp\!\left(
		\alpha\frac{
			R^{\mathrm{hh}}_{i,j}+R^{\mathrm{hh}}_{j,i}
		}{2}
		\right),
		\qquad i\neq j,
		\qquad [\mathbf W]_{i,i}=0,
		\label{eq:W}
	\end{equation}
	where $\alpha>0$ scales the rate-based coupling so that
	the exponential argument is dimensionless.
	The coupling metric is evaluated using the associations
	and private rates of the current design. 
	Let \(\mathbf{1}\) denote the all-ones vector. Define \(\mathbf{D}=\mathrm{diag}(\mathbf{W}\mathbf{1})\) and the combinatorial Laplacian \(\mathbf{L}=\mathbf{D}-\mathbf{W}\). The layout regularizer is
	\begin{equation}
		J_{\mathrm{l}} \;=\; \mathrm{tr}\!\big(\mathbf{Q}^\mathsf{T}\mathbf{L}\mathbf{Q}\big),
		\label{eq:J_layout_revised}
	\end{equation}
	which penalizes large distances between strongly connected HAPS pairs. The detailed derivation is provided in Appendix~\ref{app:layout_proof}.
	To normalize the edge weights to $[0,1]$, we use
	\begin{equation}
		\label{eq:W_tilde}
		\tilde W_{i,j} = \frac{[\mathbf{W}]_{i,j}}{\sum_{i'<j'} [\mathbf{W}]_{i',j'}},\qquad i\neq j,\quad \tilde W_{i,i}=0.
	\end{equation}
	
	Then, normalize the weighted squared distances by the maximum feasible separation \(d_{\max}\):
	\begin{equation}
		\tilde J_{\mathrm{l}} = \frac{\sum_{i<j} \tilde W_{i,j}\,\|\mathbf{q}_i-\mathbf{q}_j\|_2^2}{d_{\max}^2} \;\in\; [0,1].
	\end{equation}
	
	\paragraph{Group (Tag) Concentration Regularizer}
	
	We consider hard associations $a_{u,k}\in\{0,1\}$
	and soft association weights $\pi_{u,k}\in[0,1]$.
	The regularizer discourages concentration of
	group associations on a small number of HAPSs.
	
	\textit{Hard association.}
	Let
	$N_g^{\mathrm{srv}}\triangleq
	|\mathcal{U}_g^{\mathrm{srv}}|$
	denote the number of served users in group $g$.
	For $N_g^{\mathrm{srv}}>0$, define
	\begin{equation}
		p_{g,k}
		\triangleq
		\frac{1}{N_g^{\mathrm{srv}}}
		\sum_{u\in\mathcal{U}_g^{\mathrm{srv}}}
		\frac{a_{u,k}}{\sum_{k'\in\mathcal{K}}a_{u,k'}}.
		\label{eq:p_g_k_def}
	\end{equation}
	Thus, $p_{g,k}\geq0$ and $\sum_k p_{g,k}=1$.
	This definition covers both single-connect
	and multi-connect association.
	
	For an active group, the concentration index is
	\begin{equation}
		\phi_g\triangleq\sum_{k\in\mathcal{K}}p_{g,k}^2.
	\end{equation}
	We set $\phi_g=1$ for a group with no served users
	so that dropping the entire group does not
	artificially improve this regularizer.
	The unnormalized penalty is
	\begin{equation}
		J_{\mathrm{tag}}
		\triangleq
		\sum_{g\in\mathcal{G}}\phi_g.
		\label{eq:J_tag_revised}
	\end{equation}
	
	For $|\mathcal{K}|>1$, the normalized penalty is
	\begin{equation}
		\bar{J}_{\mathrm{tag}}
		\triangleq
		\sum_{g\in\mathcal{G}}
		\frac{\phi_g-1/|\mathcal{K}|}
		{1-1/|\mathcal{K}|}.
	\end{equation}
	Each group contributes a value in $[0,1]$.
	For $|\mathcal{K}|=1$, we set
	$\bar{J}_{\mathrm{tag}}=|\mathcal{G}|$,
	since no spatial distribution across HAPSs is possible.
	
	\textit{Soft association.}
	For users with $\sum_k\pi_{u,k}>0$, replace
	$a_{u,k}$ by $\pi_{u,k}$ in \eqref{eq:p_g_k_def},
	and define the served set and its cardinality
	using the same positivity condition.
	The definitions of $\phi_g$, $J_{\mathrm{tag}}$,
	and $\bar{J}_{\mathrm{tag}}$ remain unchanged.
	
	The penalty balances each group's association shares across HAPSs; it does not guarantee per-user connectivity after HAPS failure or credential security after compromise.
	
	We now formalize the constrained optimization problem. 
	Due to limited radio resources, some users may remain unserved in a given time slot. 
	Accordingly, the network coverage ratio is defined as
	$
	\xi_{\mathrm{cov}} \triangleq 
	\frac{|\mathcal{U}^{\mathrm{srv}}|}{|\mathcal{U}|} \in [0,1]
	$.
	
	Based on the definitions, the problem is formulated as
	\begin{equation}
		\footnotesize
		\begin{aligned}
			\min_{\boldsymbol{\Omega}}\quad 
			& \tilde{J}_{\mathrm{l}} 
			+ \lambda_{\mathrm{tag}} \bar{J}_{\mathrm{tag}}
			- \eta_J \mathcal{J}(\mathbf{R}) \\
			\text{s.t.}\quad 
			& \text{C1: } \|\dot{\mathbf{q}}_k\|_2 \le v_{\max},
			\quad \forall k \in \mathcal{K}, \\
			& \text{C2: } \|\mathbf{q}_i - \mathbf{q}_j\|_2 \ge d_{\min},
			\quad \forall i \neq j,\ i,j \in \mathcal{K}, \\
			& \text{C3: }
			\Pr\!\left(
			\gamma_{\mathrm{e},u}^{\mathrm c}
			>
			\gamma_e^{\mathrm{th,c}}
			\right)
			\leq
			\epsilon_c,
			\quad
			\forall u\in\mathcal{U}^{\mathrm{srv}}, \\
			& \text{C4: }
			\Pr\!\left(
			\max_{e\in\mathcal{E}}\gamma_{e,u}
			>
			\gamma_e^{\mathrm{th,nc}}
			\right)
			\leq
			\epsilon_{nc},
			\quad
			\forall u\in\mathcal{U}^{\mathrm{srv}}, \\		
			& \text{C5: } R_{\mathrm{p},u}\ge R_{\min},
			\qquad
			\forall u\in\mathcal{U}^{\mathrm{srv}}, \\		
			& \text{C6: } \sum_{u \in \mathcal{U}} a_{u,k} \le U_{\mathrm{perRB}},
			\quad \forall k \in \mathcal{K}, \\		
			& \text{C7: } 
			0 \le \rho_{\mathrm{b},k}, \rho_{\mathrm{c},k,g}, \rho_{\mathrm{d},k} \le 1,
			\quad \forall k,g, \\		
			& \text{C8: } 
			\xi_{\mathrm{cov}}
			\ge \xi_{\mathrm{req}} , \\
			& \text{C9: }
			|\mathcal{U}_k|\rho_{\mathrm{b},k}^2
			+
			\sum_{g\in\mathcal{G}_k}
			\rho_{\mathrm{c},k,g}^2
			+
			|\mathcal{G}_k|\rho_{\mathrm{d},k}^2
			\leq 1,
			\quad
			\forall k\in\mathcal{K}, \\
			& \text{C10: }
			\gamma_{\mathrm{c},u}
			\geq\gamma_{\mathrm{c}}^{\min},
			\qquad
			\forall u\in\mathcal{U}^{\mathrm{srv}}.
		\end{aligned}
		\label{eq:problem}
	\end{equation}
	
	Here, the optimization variable set is defined as
	$
	\boldsymbol{\Omega} \triangleq
	\big(
	\{\mathbf{q}_k\},
	\{a_{u,k}\},
	\{\rho_{\mathrm{b},k}\},
	\{\rho_{\mathrm{c},k,g}\},
	\{\rho_{\mathrm{d},k}\}
	\big)
	$.
	The nonnegative coefficients $\lambda_{\mathrm{tag}}$
	and $\eta_J$ weight the group-concentration penalty
	and private-rate fairness, respectively.
	The parameter $\xi_{\mathrm{req}}\in(0,1]$ specifies
	the minimum required coverage ratio.
	The positive thresholds
	$\gamma_e^{\mathrm{th,c}}$ and
	$\gamma_e^{\mathrm{th,nc}}$ specify the maximum
	acceptable eavesdropper SINRs in the colluding and
	non-colluding models.
	The parameters $\epsilon_c,\epsilon_{nc}\in[0,1]$
	specify the corresponding maximum violation probabilities.
	
	We use a discrete-time mobility model with sampling
	interval $\Delta t>0$.
	The vector $\dot{\mathbf q}_k$ denotes the velocity
	of HAPS $k$, and $v_{\max}>0$ is its maximum permitted
	speed.
	The parameter $d_{\min}>0$ specifies the minimum
	separation between distinct HAPSs, and $R_{\min}>0$
	is the required private-stream rate of each served user.
	
	C1--C2 constrain mobility and separation; C3--C4 bound eavesdropper-SINR exceedance probabilities; C5--C6 impose private-rate and association-capacity requirements; and C7--C9 enforce coefficient bounds, coverage, and total power. In C10, $\gamma_{\mathrm{c}}^{\min}>0$ ensures usable tag-stream SINR for served users without constituting an authentication-security guarantee.
	
	\section{SAFA-MZ Framework and Algorithms}
	\label{sec:solution}
	
	Problem~\eqref{eq:problem} is a non-convex mixed-integer program with coupled placement, association, and power decisions and stochastic secrecy constraints. We use derivative-free RCEM and graph RL enhanced by a graph autoencoder (GAE). Both obtain precoders from the RZF and MRT rules in Section~\ref{sec:precoders}.
	
	We enforce the probabilistic secrecy constraints C3 and C4
	using a Monte Carlo (MC) approximation. For a decision tuple
	$\boldsymbol{\Omega}$, define the worst-user eavesdropper
	SINRs as
	\begin{align}
		Z_c(\boldsymbol{\Omega})
		&\triangleq
		\max_{u\in\mathcal{U}^{\mathrm{srv}}}
		\gamma_{\mathrm{e},u}^{\mathrm c}(\boldsymbol{\Omega}),
		\\
		Z_{nc}(\boldsymbol{\Omega})
		&\triangleq
		\max_{u\in\mathcal{U}^{\mathrm{srv}}}
		\max_{e\in\mathcal{E}}\gamma_{e,u}.
	\end{align}
	The corresponding network-level violation probabilities are
	\begin{align}
		p_c(\boldsymbol{\Omega})
		&\triangleq
		\Pr\!\left(
		Z_c(\boldsymbol{\Omega})
		>
		\gamma_e^{\mathrm{th,c}}
		\right),
		\\
		p_{nc}(\boldsymbol{\Omega})
		&\triangleq
		\Pr\!\left(
		Z_{nc}(\boldsymbol{\Omega})
		>
		\gamma_e^{\mathrm{th,nc}}
		\right).
	\end{align}
	The worst-user chance constraints conservatively control all individual served-user constraints.
	
	Using $N_{\mathrm{mc}}$ independent channel realizations, the
	violation probabilities are estimated as
	\begin{align}
		\hat p_c(\boldsymbol{\Omega})
		&\triangleq
		\frac{1}{N_{\mathrm{mc}}}
		\sum_{n=1}^{N_{\mathrm{mc}}}
		\mathbb{I}\!\left\{
		Z_c^{(n)}(\boldsymbol{\Omega})
		>
		\gamma_e^{\mathrm{th,c}}
		\right\},
		\\
		\hat p_{nc}(\boldsymbol{\Omega})
		&\triangleq
		\frac{1}{N_{\mathrm{mc}}}
		\sum_{n=1}^{N_{\mathrm{mc}}}
		\mathbb{I}\!\left\{
		Z_{nc}^{(n)}(\boldsymbol{\Omega})
		>
		\gamma_e^{\mathrm{th,nc}}
		\right\}.
	\end{align}
	Hence, C3 and C4 are approximated by
	$\hat p_c(\boldsymbol{\Omega})\leq\epsilon_c$ and
	$\hat p_{nc}(\boldsymbol{\Omega})\leq\epsilon_{nc}$,
	respectively. Here, $Z_c^{(n)}$ and $Z_{nc}^{(n)}$ are evaluated on realization $n$, holding decisions fixed and recomputing channel-dependent precoders as in Section~\ref{sec:precoders}. Empty served-set maxima are zero, but such candidates violate C8 because $\xi_{\mathrm{req}}>0$. These empirical constraints provide no finite-sample guarantee on true violation probabilities.
	
	RCEM admits only feasible candidates to elite selection.
	GA2C uses the penalized reward during training, while
	feasibility rejection and resampling are applied only
	during deployment, with at most $N_{\mathrm{try}}$ attempts
	per state.
	The contribution is system integration; both algorithms
	serve as existing solvers.
	
	\subsection{RCEM}
	
	RCEM solves each static instance of \eqref{eq:problem}
	by adapting its sampling distribution. Let $N_s\in\mathbb{Z}_+$,
	$N_e\in\{1,\ldots,N_s\}$, and $N_{\mathrm{iter}}\in\mathbb{Z}_+$
	denote the candidates per iteration, elite count, and
	maximum iterations, respectively.
	Let $D$ denote the dimension of the sampling vector, and let
	$\boldsymbol{\theta}\in\mathbb{R}^{D}$ contain the HAPS
	coordinates, power-allocation variables, and real-valued
	association logits. At iteration $t$, the candidates are
	sampled from a diagonal Gaussian distribution with mean
	$\boldsymbol{\mu}^{(t)}\in\mathbb R^D$ and coordinate-wise
	standard deviations $\boldsymbol{\sigma}_r^{(t)}\in\mathbb R_+^D$:
	\begin{equation}
		\boldsymbol{\theta}_i^{(t)}
		\sim
		\mathcal{N}\!\left(
		\boldsymbol{\mu}^{(t)},
		\operatorname{diag}\big(
		(\boldsymbol{\sigma}_r^{(t)})^{\odot2}
		\big)
		\right),
		\qquad i=1,\ldots,N_s.
	\end{equation}
	
	Each sampled vector is mapped to a decision tuple
	$\boldsymbol{\Omega}_i^{(t)}$. HAPS coordinates are projected
	onto the admissible placement and mobility region. Here,
	$\mathcal Q\subset\mathbb R^3$ denotes the deployment box whose
	horizontal and altitude limits are given in
	Table~\ref{tab:sim_params_full}. The mobility region is the ball
	of radius $v_{\max}\Delta t$ centered at the known previous
	HAPS position; consequently, C1 uses
	$\dot{\mathbf q}_k=(\mathbf q_k-\mathbf q_k^{\mathrm{prev}})/\Delta t$,
	where $\mathbf q_k^{\mathrm{prev}}$ is that previous position.
	Projection is onto the intersection of these two regions.
	Power coefficients are clipped to $[0,1]$, and association
	logits are thresholded to obtain binary $a_{u,k}$. If C6 is
	violated, only the $U_{\mathrm{perRB}}$ users with the
	largest association scores are retained at the corresponding
	HAPS. Unserved users are then greedily associated with an
	available HAPS until C8 is satisfied or no capacity remains.

	For every HAPS, define
	\begin{equation}
		s_k
		\triangleq
		|\mathcal{U}_k|\rho_{\mathrm{b},k}^2
		+
		\sum_{g\in\mathcal{G}_k}
		\rho_{\mathrm{c},k,g}^2
		+
		|\mathcal{G}_k|\rho_{\mathrm{d},k}^2.
	\end{equation}
	If $s_k>1$, all power-allocation coefficients of HAPS $k$
	are divided by $\sqrt{s_k}$, thereby enforcing C9 while
	preserving their relative proportions.
	
	The objective value is $J(\boldsymbol{\Omega})
	\triangleq
	\tilde J_{\mathrm l}
	+\lambda_{\mathrm{tag}}\bar J_{\mathrm{tag}}
	-\eta_J\mathcal{J}(\mathbf R)$.
	Constraints are incorporated through the aggregate violation
	measure
	\begin{align}
		\mathcal V(\boldsymbol{\Omega})
		\triangleq\;&
		\lambda_{\mathrm{pen},1}
		\sum_{i<j}
		\left[
		d_{\min}-\|\mathbf q_i-\mathbf q_j\|_2
		\right]^+
		\nonumber\\
		&+
		\lambda_{\mathrm{pen},2}
		\sum_{u\in\mathcal{U}^{\mathrm{srv}}}
		\left[R_{\min}-R_{\mathrm p,u}\right]^+ \nonumber\\
		&+
		\lambda_{\mathrm{pen},3}
		\left[\xi_{\mathrm{req}}-\xi_{\mathrm{cov}}\right]^+
		\nonumber\\
		&+
		\lambda_{\mathrm{pen},4}
		\left[\hat p_c(\boldsymbol{\Omega})-\epsilon_c\right]^+
		+
		\lambda_{\mathrm{pen},5}
		\left[\hat p_{nc}(\boldsymbol{\Omega})-\epsilon_{nc}\right]^+ \nonumber\\
		&+
		\lambda_{\mathrm{pen},6}
		\sum_{u\in\mathcal{U}^{\mathrm{srv}}}
		\left[
		\gamma_{\mathrm{c}}^{\min}
		-\gamma_{\mathrm{c},u}
		\right]^+.
	\end{align}
	For $|\mathcal K|=1$, we set $\tilde J_{\mathrm l}=0$.
	
	The penalized score used for elite selection is
	$
	J_{\mathrm{pen}}(\boldsymbol{\Omega})
	=
	J(\boldsymbol{\Omega})
	+
	\mathcal V(\boldsymbol{\Omega})
	$,
	where
	\(
	\lambda_{\mathrm{pen},1},
	\lambda_{\mathrm{pen},2},
	...,
	\lambda_{\mathrm{pen},6}>0.
	\)
	Penalties depend on units and coefficients and do not
	guarantee feasibility. Each candidate allows at most
	$N_{\mathrm{try}}$ attempts; infeasible repaired candidates
	are rejected. With at least $N_e$ feasible candidates,
	the $N_e$ lowest-score candidates form
	$\mathcal{I}_{\mathrm{elite}}^{(t)}$; otherwise, the distribution
	is unchanged. Its parameters are updated as
	\begin{align}
		\boldsymbol{\mu}^{(t+1)}
		&=
		\frac{1}{N_e}
		\sum_{i\in\mathcal{I}_{\mathrm{elite}}^{(t)}}
		\boldsymbol{\theta}_i^{(t)},
		\\
		\boldsymbol{\sigma}_r^{(t+1)}
		&=
		\sqrt{
			\frac{1}{N_e}
			\sum_{i\in\mathcal{I}_{\mathrm{elite}}^{(t)}}
			\left(
			\boldsymbol{\theta}_i^{(t)}
			-\boldsymbol{\mu}^{(t+1)}
			\right)^{\odot2}}
		+\epsilon_{\mathrm{RCEM}}\mathbf{1}_D,
	\end{align}
	where $\epsilon_{\mathrm{RCEM}}>0$ prevents variance collapse. The square root is applied element-wise, and
	$\mathbf{1}_D$ is the $D$-dimensional all-ones vector. RCEM returns the best repaired candidate satisfying the deterministic and MC-approximated constraints, rather than the final distribution mean, or reports search failure if none exists. It is rerun for each new snapshot.
	
	\begin{algorithm}[t]
		\caption{RCEM}
		\label{alg:rcem_revised}
		\scriptsize
		\begin{algorithmic}[1]
			\Require $N_s,N_e,N_{\mathrm{iter}},N_{\mathrm{mc}},
			N_{\mathrm{try}},\boldsymbol{\mu}^{(0)},
			\boldsymbol{\sigma}_r^{(0)}$
			\Ensure Best repaired feasible decision $\boldsymbol{\Omega}^{\star}$
			
			\State Initialize sampling parameters; set the feasible incumbent to empty
			
			\For{$t=0$ to $N_{\mathrm{iter}}-1$}
			\State $\mathcal{I}_{\mathrm{feas}}^{(t)}\gets\emptyset$
			
			\For{$i=1$ to $N_s$}
			\For{$r=1$ to $N_{\mathrm{try}}$}
			\State Sample
			$\boldsymbol{\theta}_i\sim
			\mathcal{N}\!\left(
			\boldsymbol{\mu},
			\operatorname{diag}
			(\boldsymbol{\sigma}_r^{\odot 2})
			\right)$
			\State Map $\boldsymbol{\theta}_i$ to
			$\boldsymbol{\Omega}_i$
			\State Project positions and clip power variables
			\State Binarize association logits
			\State Repair C6, C8, and C9, as applicable
			\State Compute rates, SSEs, and fairness
			\State Estimate $\hat p_c(\boldsymbol{\Omega}_i)$
			and $\hat p_{nc}(\boldsymbol{\Omega}_i)$
			\State Evaluate $J_{\mathrm{pen}}(\boldsymbol{\Omega}_i)$
			
			\If{$\boldsymbol{\Omega}_i$ satisfies all evaluated constraints}
			\State $\mathcal{I}_{\mathrm{feas}}^{(t)}
			\gets\mathcal{I}_{\mathrm{feas}}^{(t)}\cup\{i\}$
			\State Update $\boldsymbol{\Omega}^{\star}$ if empty
			or if $\boldsymbol{\Omega}_i$ improves $J$
			\State \textbf{break}
			\EndIf
			\EndFor
			\EndFor
			
			\If{$|\mathcal{I}_{\mathrm{feas}}^{(t)}|\geq N_e$}
			\State Select the $N_e$ feasible candidates
			\State Update $\boldsymbol{\mu}$ and
			$\boldsymbol{\sigma}_r$ from the elite samples
			\EndIf
			\EndFor
			
			\State \Return $\boldsymbol{\Omega}^{\star}$ if available;
			otherwise report search failure
		\end{algorithmic}
	\end{algorithm}

	\subsection{GA2C}
	
	GA2C fuses a permutation-invariant GAE representation of HAPS interactions with environment observations to learn joint decisions $\{\mathbf{q}_k\}$, $\{a_{u,k}\}$, and $\{\rho_{\mathrm{b},k},\rho_{\mathrm{c},k,g},\rho_{\mathrm{d},k}\}$.

	At step $t$, the state is $\mathcal{S}_t\triangleq(\mathbf{o}_t,\mathcal{G}_t)$ with HAPS graph $\mathcal{G}_t\triangleq(\mathbf{X}_t,\mathbf{W}_t)$. The observation $\mathbf{o}_t$ includes HAPS positions, associations, power coefficients, and large-scale channel statistics. User locations, tag identities, and HAPS--user channel features feed a separate actor encoder for association decisions. Node features $\mathbf{X}_t\in\mathbb{R}^{|\mathcal{K}|\times d_x}$ contain positions, served-user counts, and power indicators; $d_x$ is the feature dimension. The symmetric zero-diagonal matrix $\mathbf{W}_t\in\mathbb{R}^{|\mathcal{K}|\times|\mathcal{K}|}$ uses the coupling weights in~\eqref{eq:W}. Appendix~\ref{app:perm_priv} establishes graph-encoder equivariance and mean-pooling invariance, but not full-actor equivariance or generalization to unseen network sizes.
	
	Let $\breve{\mathbf{W}}_t$ be the element-wise min--max normalization of $\mathbf{W}_t$ into $[0,1]$ (for stable reconstruction loss), with zero diagonal. 
	The symmetrically normalized adjacency is defined as
	\begin{equation}
		\hat{\mathbf{A}}_t
		=
		\mathbf{D}_t^{-1/2}(\breve{\mathbf{W}}_t+\mathbf{I}_{|\mathcal{K}|})\mathbf{D}_t^{-1/2},
		\quad
		\mathbf{D}_t=\mathrm{diag}((\breve{\mathbf{W}}_t+\mathbf{I}_{|\mathcal{K}|})\mathbf{1}),
	\end{equation}
	where $\hat{\mathbf{A}}_t$ is the symmetrically normalized adjacency used as the GNN propagation operator and $\mathbf{1}\in\mathbb{R}^{|\mathcal{K}|}$ is the all-ones vector.
	The GNN encoder follows the graph convolutional network (GCN) propagation rule
	\begin{equation}
		\mathbf{H}_{t}^{(\ell+1)}
		=
		\sigma_f\!\left(
		\hat{\mathbf{A}}_t
		\mathbf{H}_{t}^{(\ell)}
		\boldsymbol{\Theta}^{(\ell)}
		\right),
		\qquad
		\mathbf{H}_{t}^{(0)}=\mathbf{X}_t,
		\label{eq:GNN_Enc}
	\end{equation}
	where
	$\mathbf{H}_{t}^{(\ell)}
	\in\mathbb{R}^{|\mathcal{K}|\times d_\ell}$
	is the node-embedding matrix,
	$\boldsymbol{\Theta}^{(\ell)}
	\in\mathbb{R}^{d_\ell\times d_{\ell+1}}$
	is a trainable transformation matrix, and
	$\sigma_f(\cdot)$ is an element-wise activation function. Here, $L$ is the
	number of graph-convolutional layers, $d_0=d_x$, and $d_L=d_z$.
	
	Let
	$\mathbf{Z}_t\triangleq\mathbf{H}_{t}^{(L)}
	\in\mathbb{R}^{|\mathcal{K}|\times d_z}$
	denote the final node-embedding matrix, and let
	$\mathbf{z}_{\mathrm{gr}}\in\mathbb{R}^{1\times d_z}$ denote the
	mean-pooled graph representation:
	\begin{equation}
		\mathbf{z}_{\mathrm{gr}}
		\triangleq
		\frac{1}{|\mathcal{K}|}
		\sum_{k\in\mathcal{K}}\mathbf{Z}_t(k,:).
	\end{equation}
	The actor and critic receive $\mathbf{z}_{\mathrm{gr}}$.
	The decoder reconstructs normalized edges from node embeddings:
	
	\begin{equation}
		\hat W_{t,i,j}
		\triangleq
		\operatorname{sigmoid}\!\left(
		\mathbf{Z}_t(i,:)\mathbf{Z}_t(j,:)^{\mathsf T}
		\right).
	\end{equation}
	The adjacency reconstruction loss is
	\begin{align}
		&\mathcal{L}_{\mathrm{adj},t}
		= \nonumber\\
		&-\sum_{i,j}
		\left[
		\breve W_{t,i,j}\log\hat W_{t,i,j}
		+
		(1-\breve W_{t,i,j})
		\log(1-\hat W_{t,i,j})
		\right].
	\end{align}
	
	The mean-pooled node-feature vector is denoted by
	\begin{equation}
		\bar{\mathbf X}_t
		\triangleq
		\frac{1}{|\mathcal K|}
		\sum_{k\in\mathcal K}\mathbf X_t(k,:).
	\end{equation}
	Here, $\bar{\mathbf X}_t\in\mathbb{R}^{1\times d_x}$
	is the pooled feature vector, and
	$\hat{\mathbf X}_t\in\mathbb{R}^{1\times d_x}$
	denotes its reconstruction from $\mathbf z_{\mathrm{gr}}$. The auxiliary graph-autoencoder
	loss is
	\begin{equation}
		\mathcal{L}_{\mathrm{GAE},t}
		=
		\alpha_{\mathrm{adj}}\mathcal{L}_{\mathrm{adj},t}
		+
		\alpha_x
		\left\|
		\bar{\mathbf X}_t-\hat{\mathbf X}_t
		\right\|_2^2,
	\end{equation}
	where $\alpha_{\mathrm{adj}}\geq0$ and $\alpha_x\geq0$ weight the auxiliary reconstruction tasks. With the adopted final-layer ReLU, all embedding inner
	products are nonnegative, implying $\hat W_{t,i,j}\geq1/2$.
	The decoder therefore cannot exactly reconstruct target
	weights below $1/2$; its loss serves only as an auxiliary
	representation-learning objective.

	We adopt an A2C architecture. The actor outputs Gaussian policies for continuous actions (HAPS motion and power coefficients) and Bernoulli policies for binary associations $\{a_{u,k}\}$ derived from user-dependent logits.
	The critic estimates the state value $V(\mathcal{S}_t)$. The temporal difference (TD) target and advantage are
	\begin{equation}
		\label{eq:TD_target and advantage}
		y_t = r_t + \gamma_{\mathrm{rl}} V(\mathcal{S}_{t+1}), \qquad
		A_t = y_t - V(\mathcal{S}_t),
	\end{equation}
	where $\gamma_{\mathrm{rl}}\in[0,1)$ is the discount factor.
	Terminal bootstrap values are zero. Value targets and
	actor advantages are detached during differentiation.
	Let $\mathbf{a}_t$ be the joint motion, power, and association
	action at time $t$, with policy
	$\pi(\mathbf{a}_t|\mathcal{S}_t)$. The A2C loss is
	\begin{equation}
		\mathcal{L}_{\mathrm{A2C}}
		=
		-\log \pi(\mathbf{a}_t|\mathcal{S}_t)\, A_t
		+ c_v \big(y_t - V(\mathcal{S}_t)\big)^2
		- c_e \mathcal{H}\!\big(\pi(\cdot|\mathcal{S}_t)\big),
	\end{equation}
	where $\mathcal{H}(\pi(\cdot|\mathcal{S}_t))$ is the policy entropy, and $c_v\ge 0$, $c_e\ge 0$ are the value-loss and entropy-regularization coefficients.
	
	Let $\boldsymbol{\Omega}_t$ denote the decision tuple induced by action $\mathbf{a}_t$ at state $\mathcal{S}_t$. 
	The action mapping clips the HAPS coordinates to
	$\mathcal Q$, bounds the displacement according to C1, and
	bounds the power coefficients according to C7. The policy likelihood is evaluated for the sampled action before deterministic clipping and repair. To align
	training with the remaining constraints in
	\eqref{eq:problem}, the instantaneous reward is
	\begin{align}
		r_t
		={}&-
		\left(
		\tilde J_{\mathrm l}
		+\lambda_{\mathrm{tag}}\bar J_{\mathrm{tag}}
		-\eta_J\mathcal J(\mathbf R)
		\right)
		\nonumber\\
		&-
		\lambda_d
		\sum_{i<j}
		\left[
		d_{\min}-\|\mathbf q_i-\mathbf q_j\|_2
		\right]^+
		\nonumber\\
		&-
		\lambda_R
		\sum_{u\in\mathcal U^{\mathrm{srv}}}
		\left[R_{\min}-R_{\mathrm p,u}\right]^+
		-
		\lambda_\xi
		\left[\xi_{\mathrm{req}}-\xi_{\mathrm{cov}}\right]^+
		\nonumber\\
		&-
		\lambda_c
		\left[
		\hat p_c(\boldsymbol{\Omega}_t)-\epsilon_c
		\right]^+
		-
		\lambda_{nc}
		\left[
		\hat p_{nc}(\boldsymbol{\Omega}_t)-\epsilon_{nc}
		\right]^+
		\nonumber\\
		&-
		\lambda_U
		\sum_{k\in\mathcal K}
		\left[
		\sum_{u\in\mathcal U}a_{u,k}
		-U_{\mathrm{perRB}}
		\right]^+
		\nonumber\\
		&-
		\lambda_P
		\sum_{k\in\mathcal K}
		\left[
		|\mathcal U_k|\rho_{\mathrm b,k}^2
		+
		\sum_{g\in\mathcal G_k}\rho_{\mathrm c,k,g}^2
		+
		|\mathcal G_k|\rho_{\mathrm d,k}^2
		-1
		\right]^+ \nonumber \\
		&-\lambda_{\mathrm{auth}}
		\sum_{u\in\mathcal{U}^{\mathrm{srv}}}
		\left[
		\gamma_{\mathrm{c}}^{\min}
		-\gamma_{\mathrm{c},u}
		\right]^+.
		\label{eq:GA2C_reward}
	\end{align}
	Here, all penalty coefficients are nonnegative. The terms weighted by $\lambda_d$, $\lambda_R$, $\lambda_\xi$,
	$\lambda_c$, $\lambda_{nc}$, $\lambda_U$, $\lambda_P$, and $\lambda_{\mathrm{auth}}\geq0$
	correspond to C2, C5, C8, C3, C4, C6, C9, and C10 respectively.
	Finite penalties do not guarantee feasibility. GA2C
	maximizes discounted reward based on \eqref{eq:problem};
	RCEM optimizes individual snapshots.
	
	The overall training loss is
	$
		\mathcal L
		=
		\mathcal L_{\mathrm{A2C}}
		+
		\lambda_{\mathrm{GAE}}
		\mathcal L_{\mathrm{GAE}}
	$,
	where $\lambda_{\mathrm{GAE}}\geq0$ is the auxiliary-loss
	coefficient.
	Parameters are updated using Adam with gradient clipping.
	For brevity, the time index on $\mathcal L_{\mathrm{GAE},t}$
	is omitted in the overall loss. Let $\boldsymbol\omega$,
	$\boldsymbol\psi$, and $\boldsymbol\phi$ denote the actor,
	critic, and GAE parameters, respectively, and
	let $T_{\max}$ denote the maximum number of interaction steps
	in an episode.
	
	\begin{algorithm}[t]
		\caption{GA2C}
		\label{alg:GA2C_revised}
		\scriptsize
		\begin{algorithmic}[1]
			\Require $\gamma_{\mathrm{rl}},\lambda_{\mathrm{GAE}},
			\alpha_{\mathrm{adj}},\alpha_x,
			\lambda_d,\lambda_R,\lambda_\xi,\lambda_c,\lambda_{nc},
			\lambda_U,\lambda_P,\lambda_{\mathrm{auth}},
			N_{\mathrm{mc}}$
			\Ensure Trained policy $\pi_{\boldsymbol{\omega}}$,
			value function $V_{\boldsymbol{\psi}}$, and feasible
			deployment decisions or search failure
			
			\State Initialize actor $\boldsymbol{\omega}$,
			critic $\boldsymbol{\psi}$, and GAE parameters
			$\boldsymbol{\phi}$
			
			\Statex \textit{Training:}
			\For{each episode}
			\State Observe initial state
			$\mathcal{S}_0=(\mathbf{o}_0,\mathcal{G}_0)$
			\For{$t=0$ to $T_{\max}-1$}
			\State Construct normalized adjacency
			$\hat{\mathbf{A}}_t$ from $\mathcal{G}_t$
			\State Encode the HAPS graph
			$\rightarrow\mathbf{z}_{\mathrm{gr}}$
			\State Sample
			$\mathbf{a}_t\sim
			\pi(\cdot|\mathcal{S}_t,\mathbf{z}_{\mathrm{gr}})$
			\State Map and repair $\mathbf{a}_t$ to obtain
			$\boldsymbol{\Omega}_t$, and execute the mapped decision
			\State Estimate
			$\hat p_c(\boldsymbol{\Omega}_t)$ and
			$\hat p_{nc}(\boldsymbol{\Omega}_t)$ using $N_{\mathrm{mc}}$
			\State Compute the penalized reward $r_t$
			\State Observe next state $\mathcal{S}_{t+1}$
			\State Compute the TD target and advantage in
			\eqref{eq:TD_target and advantage}
			\State Update $\boldsymbol{\omega}$,
			$\boldsymbol{\psi}$, and $\boldsymbol{\phi}$
			\EndFor
			\EndFor
			
			\Statex \textit{Deployment without parameter updates:}
			\For{each deployment state $\mathcal{S}$}
			\State Construct and encode its HAPS graph
			\State Set the accepted decision to empty
			\For{$r=1$ to $N_{\mathrm{try}}$}
			\State Sample an action from $\pi_{\boldsymbol{\omega}}$
			\State Map and repair the action to obtain
			$\boldsymbol{\Omega}$
			\State Evaluate all applicable deterministic and
			MC constraints
			\If{$\boldsymbol{\Omega}$ is feasible}
			\State Accept and execute $\boldsymbol{\Omega}$
			\State \textbf{break}
			\EndIf
			\EndFor
			\If{no decision was accepted}
			\State Report search failure
			\EndIf
			\EndFor
			
			\State \Return $\boldsymbol{\omega},
			\boldsymbol{\psi},\boldsymbol{\phi}$
		\end{algorithmic}
	\end{algorithm}

	\section{Computational Complexity Analysis}
	\label{sec:complexity}
	
	We distinguish per-snapshot optimization and offline training
	from online policy inference. Let $E_{\mathrm{gr}}$ denote
	the number of undirected edges in the HAPS interaction graph.

	Let $C_{\mathrm{eval}}$ denote the operation count for evaluating
	one candidate under one channel realization, including precoding,
	useful and interfering stream coefficients, AN powers, and
	secrecy metrics. Constructing the RZF precoders costs $\mathcal{O}\!\left(
	\sum_{k\in\mathcal{K}}
	\left[M|\mathcal{U}_k|^2+|\mathcal{U}_k|^3\right]
	\right)$.
	Interference evaluation includes links from all active HAPSs
	to receivers sharing the corresponding group RB, including
	receivers not associated with the interfering HAPS.
	For the colluding model in \eqref{eq:colluding_SINR},
	a direct dense solution of
	$\mathbf{C}_u\mathbf{v}=\mathbf{f}_u$ costs
	$\mathcal{O}(|\mathcal{E}|^3)$ per served user,
	in addition to covariance construction, where $\mathbf{v}$
	is an auxiliary solution vector. Factorization reuse or
	structured implementations can reduce this cost.
	MC evaluation costs
	$\mathcal{O}(N_{\mathrm{mc}}C_{\mathrm{eval}})$
	when precoders are recomputed for each realization.

	Let $D$ denote the sampling dimension and
	$C_{\mathrm{repair}}$ the cost of mapping and repairing
	one candidate. Sorting-based association repair contributes $\mathcal{O}\!\left(
	|\mathcal{K}|\,|\mathcal{U}|\log|\mathcal{U}|
	\right)$, in addition to placement and power-allocation operations.
	With $N_s$ candidates and $N_{\mathrm{iter}}$ iterations,
	including graph construction, MC evaluation, and
	candidate sorting, the per-snapshot cost is $C_{\mathrm{RCEM}}
	=\mathcal{O}\!\Big(
	N_{\mathrm{iter}}N_s
	\big[
	N_{\mathrm{try}}
	(D+|\mathcal{K}|^2+C_{\mathrm{repair}}
	+N_{\mathrm{mc}}C_{\mathrm{eval}})
	+\log N_s
	\big]
	\Big)$.
	RCEM repeats this search for each independently optimized
	snapshot.

	For an $L$-layer graph encoder with layer widths $d_\ell$,
	the forward-pass cost is $C_{\mathrm{GNN}}
	=\mathcal{O}\!(
	\sum_{\ell=0}^{L-1}
	[(E_{\mathrm{gr}}+|\mathcal{K}|)d_\ell
	+|\mathcal{K}|d_\ell d_{\ell+1}]
	)$.
	The $|\mathcal{K}|$ term accounts for self-loops.
	Let $C_{\mathrm{actor}}$ denote the actor forward-pass cost,
	including its user-information encoder and action-output
	heads. This cost depends on the network dimensions,
	including the $|\mathcal{U}|\,|\mathcal{K}|$
	association outputs.
	Policy inference alone costs
	$\mathcal{O}(C_{\mathrm{GNN}}+C_{\mathrm{actor}})$
	and does not require the critic or adjacency decoder.
	For fixed depth and layer widths, the graph-encoder cost is
	$\mathcal{O}(|\mathcal{K}|^2)$ for a dense graph and
	$\mathcal{O}(|\mathcal{K}|)$ for a bounded-degree graph.
	
	Including graph construction, action repair, and MC
	feasibility checks, the per-snapshot cost with at most
	$N_{\mathrm{try}}$ action attempts is $C_{\mathrm{GA2C}}
	=\mathcal{O}\!\Big(
	|\mathcal{K}|^2
	+N_{\mathrm{try}}
	\big[C_{\mathrm{GNN}}+C_{\mathrm{actor}}
	+C_{\mathrm{repair}}
	+N_{\mathrm{mc}}C_{\mathrm{eval}}\big]
	\Big)$. 
	During offline training, each interaction additionally
	incurs critic, decoder, and backpropagation costs;
	the all-pairs adjacency decoder alone costs
	$\mathcal{O}(|\mathcal{K}|^2d_z)$.
	GA2C avoids repeated candidate search while retaining
	online action-processing and feasibility costs.
	
	When dense colluding-receiver solves dominate all other
	costs, the per-snapshot bounds simplify to
	$\mathcal{O}(N_{\mathrm{iter}}N_sN_{\mathrm{try}}
	N_{\mathrm{mc}}|\mathcal{U}|\,|\mathcal{E}|^3)$ for RCEM and
	$\mathcal{O}(N_{\mathrm{try}}N_{\mathrm{mc}}
	|\mathcal{U}|\,|\mathcal{E}|^3)$ for GA2C,
	excluding offline training.

	\section{Simulation Results}
	\label{sec:results}
	
	We evaluate tag- and private-stream secrecy, fairness, the effects of AN and multi-HAPS connectivity, scalability, and HAPS outages. Table~\ref{tab:sim_params_full} gives default settings; experiment-specific changes are stated below.
	
	\begin{table}[t]
		\centering
		\caption{Simulation Parameters}
		\scriptsize
		\label{tab:sim_params_full}
		\renewcommand{\arraystretch}{0.90}
		\begin{tabular}{l c}
			\hline
			\textbf{Parameter} & \textbf{Value} \\
			\hline
			
			Number of iterations (RCEM) & $50$ \\
			Number of training episodes (GA2C)  & $50$ \\
			Number of test episodes (GA2C) & $10$ \\
			Discount factor ($\gamma$) & $0.98$ \\
			Learning rate ($\eta$) & $3 \times 10^{-4}$ \\
			
			Number of HAPSs ($|\mathcal{K}|$) & $4$ \\
			Number of users ($|\mathcal{U}|$) & $20$ \\
			Number of eavesdroppers ($|\mathcal{E}|$) & $50$ \\
			Number of groups ($|\mathcal{G}|$) & $3$ \\
			Horizontal area & $x,y \in [-5000,5000]$ m \\
			HAPS altitude & $z \in [18000,22000]$ m \\
			Minimum distance ($d_{\min}$) & $2000$ m \\
			Maximum velocity ($v_{\max}$) & $50$ m/s \\
			Time step ($\Delta t$) & $1$ s \\
			
			LoS parameters ($a_l,b_l$)
			& $5.0,\;5.7296~\mathrm{rad}^{-1}$ \\
			Reference elevation ($\theta_0$)
			& $\pi/12~\mathrm{rad}$ \\
			Path-loss constant ($C_0$) & $\left(\frac{0.15}{4\pi}\right)^2$ \\
			Path-loss exponent (LoS/NLoS) & $2.0/3.5$ \\
			Shadowing std. ($\sigma_S$) & $6$ dB \\
			Number of antennas ($M$) & $64$ \\
			Rician factor ($\kappa_r$) & $10$ \\
			Wavelength & $0.15$ m \\
			Antenna spacing & $0.075$ m \\
			
			Transmit power per HAPS ($P_{\mathrm{t},k}$) & $10$ W \\
			Noise power ($\sigma^2$) & $4\times10^{-14}$ \\
			Capacity per HAPS ($U_{\mathrm{perRB}}$) & $8$ users \\
			
			\hline
		\end{tabular}
	\end{table}
	
	The GNN has $L=2$ layers with $d_0=d_x$, $d_1=64$, $d_2=d_z=32$, and rectified linear unit (ReLU) activation. The actor and critic are two- and four-layer multilayer perceptron (MLP) networks, respectively, with $256$ hidden units and $\tanh(\cdot)$ activation. ReLU reduces vanishing gradients; $\tanh$ bounds neural activations, while action mapping bounds sampled Gaussian actions. Continuous policies use fixed standard deviation $\sigma_a=0.5$; associations use Bernoulli policies. Adam uses learning rate $3\times10^{-4}$, with $c_v=0.5$, $c_e=10^{-3}$, $\lambda_{\mathrm{GAE}}=1.0$, $\alpha_{\mathrm{adj}}=1.0$, $\alpha_x=1.0$, and gradient-norm clipping to $1.0$. Table notation $\gamma$ means $\gamma_{\mathrm{rl}}$, and $\eta$ is the Adam learning rate. RCEM uses diagonal Gaussian sampling, elite size $N_e$, and stability constant $\epsilon_{\mathrm{RCEM}}>0$. HAPSs are initialized away from the deployment-box boundaries and constrained to the region in Table~\ref{tab:sim_params_full}, which keeps evaluation computationally manageable.
	The remaining parameters are set as follows: $R_{\min}=0.2$, $\xi_{\mathrm{req}}=60\%$,
	$\gamma_e^{\mathrm{th,c}}=0.1$ and
	$\gamma_e^{\mathrm{th,nc}}=0.1$,
	$\epsilon_c=0.1$ and
	$\epsilon_{nc}=0.1$,
	$N_{\mathrm{mc}}=1000$,
	$N_s=60$, $N_e=15$,
	$\epsilon_{\mathrm{RCEM}}=0.01$,	$\lambda_{\mathrm{pen},1}=10^{-3}$,
	$\lambda_{\mathrm{pen},2}=10$,
	$\lambda_{\mathrm{pen},3}=10$,
	$\lambda_{\mathrm{pen},4}=10$,
	$\lambda_{\mathrm{pen},5}=10$,
	$\lambda_{\mathrm{pen},6}=1$,
	$\lambda_d=10^{-3}$,
	$\lambda_R=10$,
	$\lambda_\xi=10$,
	$\lambda_c=10$,
	$\lambda_{nc}=10$,
	$\lambda_U=1$,
	$\lambda_P=1$,
	$\lambda_{\mathrm{auth}}=10$,
	$\alpha=0.02$,
	$\lambda_{\mathrm{tag}}=0.3$,
	$\eta_J=1$,
	$\gamma_{\mathrm{c}}^{\min}=0.1$,
	$\lambda_k=\frac{|\mathcal{U}_k|\sigma^2}{P_{\mathrm{t},k}}$,
	$\alpha_{k,u}=1$,
	$N_{\mathrm{try}}=10$,
	$T_{\max}=20$, and
	$d_x=12$.

	\subsection{SSE and Fairness}
	\label{subsec:results_haps_group_sse}
	
	In Fig.~\ref{fig:sse_haps}, GA2C improves the HAPS sum,
	mean, and minimum; the minimum does not guarantee failure
	robustness. It also improves the minimum group SSE
	(Fig.~\ref{fig:sse_group}). Figure~\ref{fig:sse_user} shows
	average per-user gains and the reduction under collusion,
	while Fig.~\ref{fig:s2c} evaluates group-averaged tag SSE.
	Checking C5 requires private rates $R_{\mathrm p,u}$.

	\begin{figure}[t]
		\centering
		\captionsetup[subfloat]{font=footnotesize}
		
		\subfloat[]{
			\includegraphics[width=0.46\columnwidth]{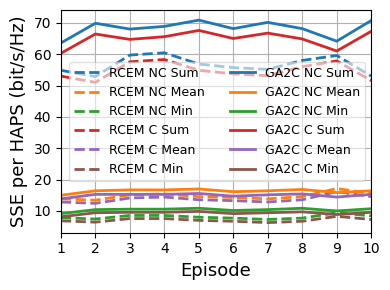}
			\label{fig:sse_haps}
		}
		\hfill
		\subfloat[]{
			\includegraphics[width=0.46\columnwidth]{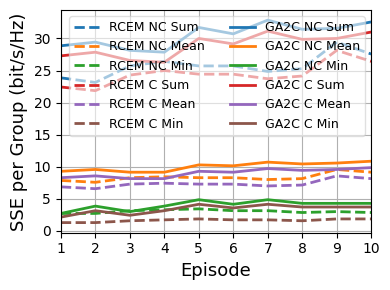}
			\label{fig:sse_group}
		}
		\\
		
		\subfloat[]{
			\includegraphics[width=0.46\columnwidth]{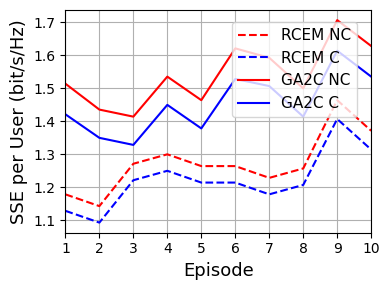}
			\label{fig:sse_user}
		}
		\hfill
		\subfloat[]{
			\includegraphics[width=0.46\columnwidth]{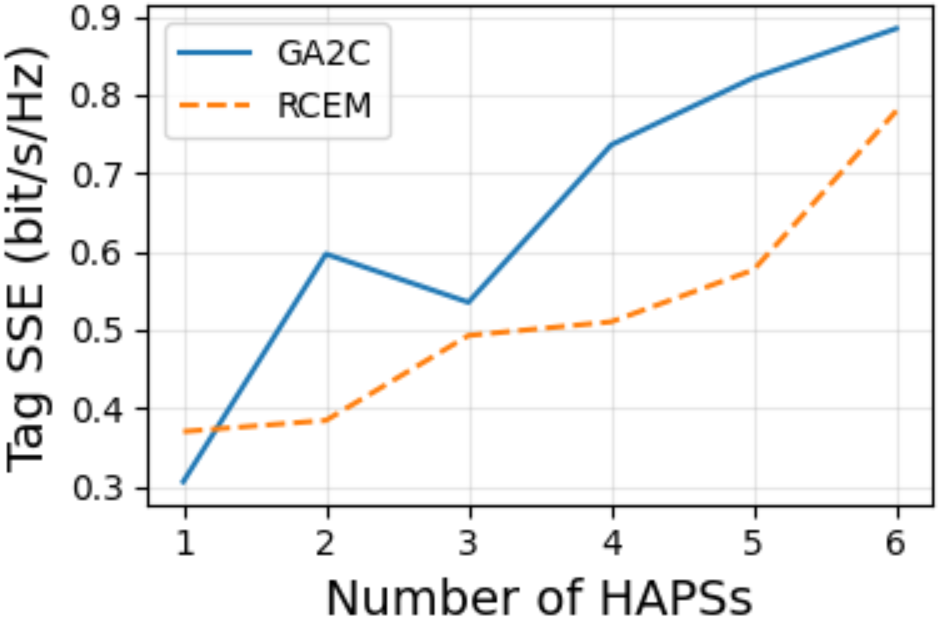}
			\label{fig:s2c}
		}
		
		\caption{SSE results: (a) HAPS sum, mean, and minimum using
			\eqref{eq:sse_haps}; (b) group sum, mean, and minimum using
			\eqref{eq:sse_group} (A moving-average smoothing filter is applied only for visual readability); (c) average per-user SSE using
			\eqref{eq:sse_average}, per test episode; and (d) illustrative
			group-averaged authentication-tag SSE using
			\eqref{eq:sse_tag_average} versus the number of HAPSs,
			with $U_{\mathrm{perRB}}=64$. C and NC denote colluding
			and non-colluding eavesdropping, respectively.}
		\label{fig:results_combined}
	\end{figure}

	Figure~\ref{fig:fairness} shows private-rate fairness $\mathcal{J}(\mathbf{R})$ based on $R_{\mathrm{p},u}$. Fairness generally increases with $|\mathcal{K}|$ in Fig.~\ref{fig:fairness_vs_haps}; GA2C performs better in the reported larger networks, while RCEM can achieve higher fairness in smaller networks through evolutionary search.
	
	\begin{figure}[t]
		\centering
		\captionsetup[subfloat]{font=footnotesize}
		
		\subfloat[]{
			\includegraphics[width=0.45\columnwidth]{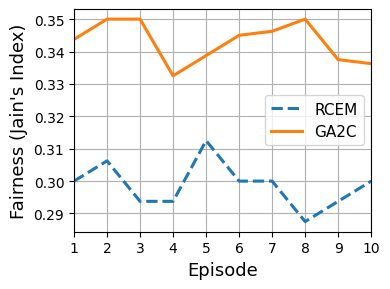}
			\label{fig:fairness}
		}
		\hfill
		\subfloat[]{
			\includegraphics[width=0.45\columnwidth]{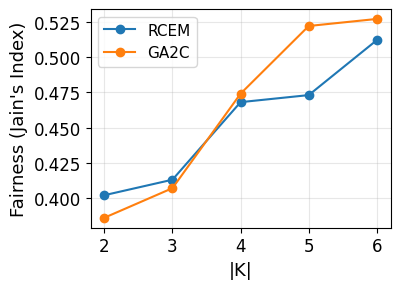}
			\label{fig:fairness_vs_haps}
		}
		
		\caption{Private-rate fairness using \eqref{eq:jain}:
			(a) evolution over 10 test episodes; (b) versus
			$|\mathcal{K}|$ with $|\mathcal{U}|=16$, $|\mathcal{G}|=4$,
			and $|\mathcal{E}|=50$.}
	\end{figure}
	
	\subsection{Comparison between Methods and Scalability Analysis}

	Figures~\ref{fig:mean10} and~\ref{fig:zoom_comp} use a separate
	run from Fig.~\ref{fig:results_combined}, with identical sets
	of 10 test episodes for both methods and their baselines.
	The \emph{w/o AN} baseline disables AN; \emph{single} baseline restricts each user
	to association with at most one HAPS.
	GA2C performs better in the reported larger networks,
	without establishing zero-shot generalization.
	Table~\ref{tab:relative_gain} compares each method with its
	own baselines: GA2C gains are 21.1--23.2\% over w/o AN
	and 135--139\% over single-connect. Absolute values in
	Fig.~\ref{fig:mean10} compare RCEM and GA2C directly.
	Figure~\ref{fig:zoom_comp} shows the benefits of AN and
	spatial coordination; Fig.~\ref{fig:3D} shows the different
	optimized positions and associations.
	
	\begin{figure}[t]
		\centering
		\includegraphics[width=\columnwidth]{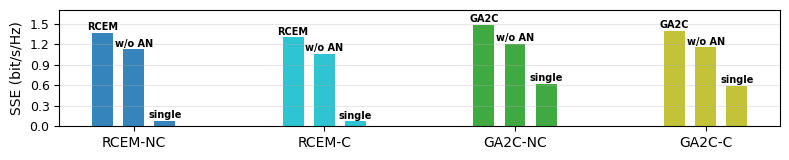}
		\caption{Average per-user private-stream SSE using
			\eqref{eq:sse_average}, averaged over the 10 test episodes
			in Fig.~\ref{fig:zoom_comp}.}
		\label{fig:mean10}
	\end{figure}
	
	\begin{table}[t]
		\centering
		\caption{SSE ratios relative to w/o AN and single,
			using the averages in Fig.~\ref{fig:mean10}.}
		\label{tab:relative_gain}
		\scriptsize
		\renewcommand{\arraystretch}{0.95}
		\begin{tabular}{lcccc}
			\hline
			& \multicolumn{2}{c}{RCEM} & \multicolumn{2}{c}{GA2C} \\
			Scenario 
			& w/o AN & single 
			& w/o AN & single \\
			\hline
			Non-Colluding & 1.222 & 17.69 & 1.232 & 2.39 \\
			Colluding     & 1.223 & 17.70 & 1.211 & 2.35 \\
			\hline
		\end{tabular}
	\end{table}
	
	\begin{figure}[t]
		\centering
		\captionsetup[subfloat]{font=footnotesize}
		
		\subfloat[RCEM]{
			\includegraphics[width=0.45\columnwidth]{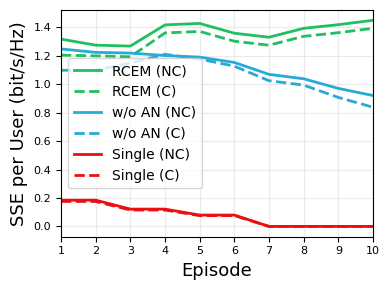}
			\label{fig:rcem_comp}
		}
		\hfill
		\subfloat[GA2C]{
			\includegraphics[width=0.45\columnwidth]{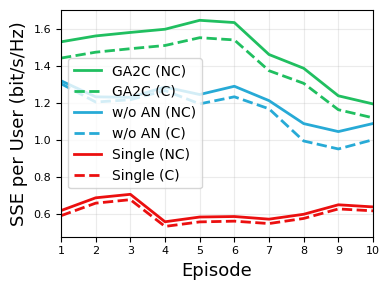}
			\label{fig:ga2c_comp}
		}
		
		\caption{Average per-user SSE using \eqref{eq:sse_average}
			in each test episode under colluding and non-colluding
			eavesdropping.}
		\label{fig:zoom_comp}
	\end{figure}
	
	\begin{figure}[t]
		\centering
		\includegraphics[width=\columnwidth]{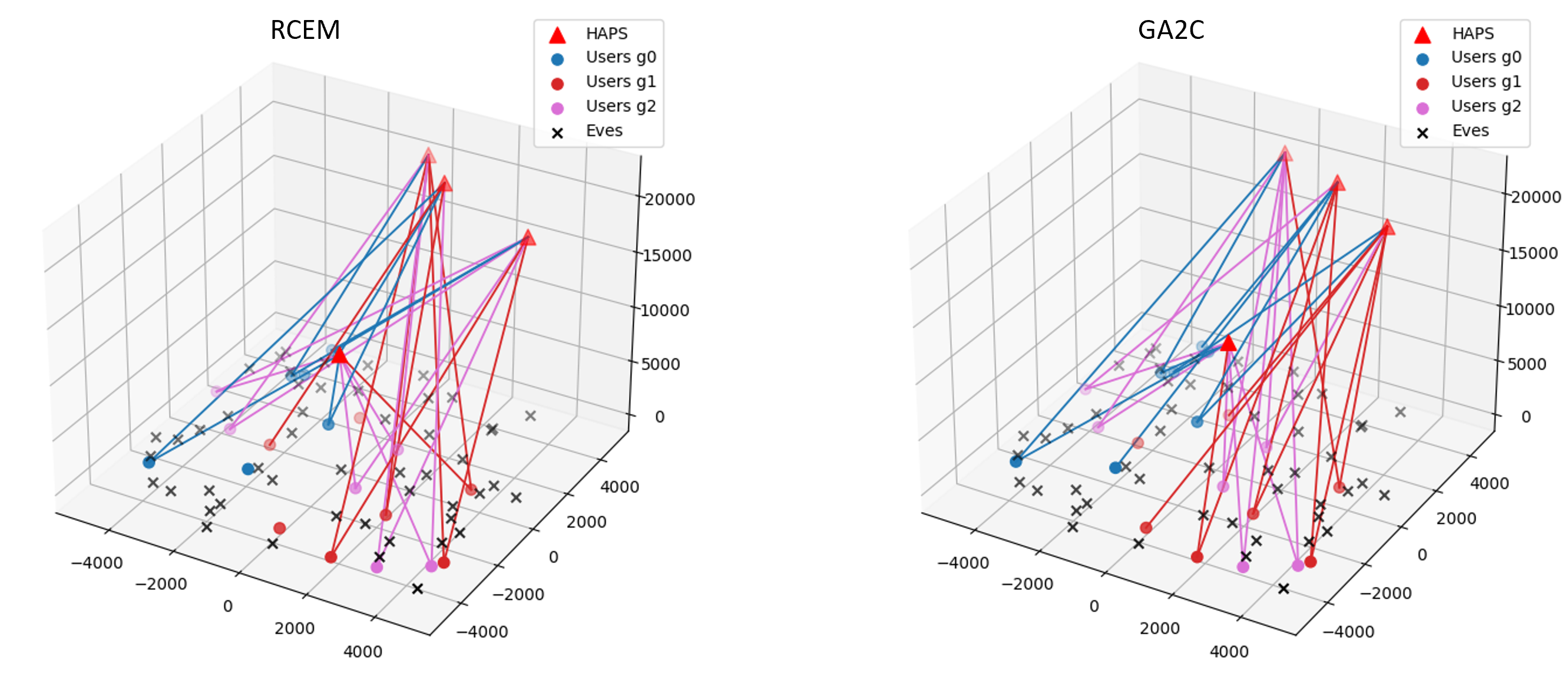}
		\caption{Comparison between the two methods.}
		\label{fig:3D}
	\end{figure}
	
	With non-colluding eavesdroppers, figure~\ref{fig:sse_vs_u} shows decreasing SSE with $|\mathcal{U}|$ as MUI and resource competition grow; GA2C maintains higher SSE in the reported larger cases. These comparisons are specific to the tested configurations and search budgets. SSE increases with $|\mathcal{K}|$ in Fig.~\ref{fig:sse_vs_k}, which also increases total network power under the per-HAPS budget. Increasing $|\mathcal{G}|$ reduces SSE in Fig.~\ref{fig:sse_vs_g} while changing power allocation and RB partition; values remain group-RB-normalized, not total-bandwidth-normalized. More eavesdroppers tighten secrecy constraints and reduce SSE (Fig.~\ref{fig:sse_vs_e}).
	
	\begin{figure}[t] 
		\centering
		\captionsetup[subfloat]{font=footnotesize}
		
		\subfloat[]{
			\includegraphics[width=0.45\columnwidth]{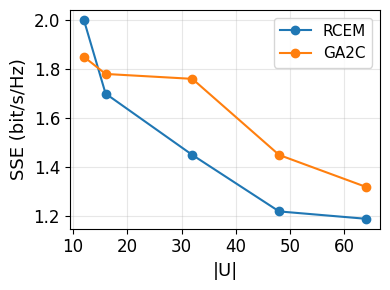}
			\label{fig:sse_vs_u}
		}
		\hfill
		\subfloat[]{
			\includegraphics[width=0.45\columnwidth]{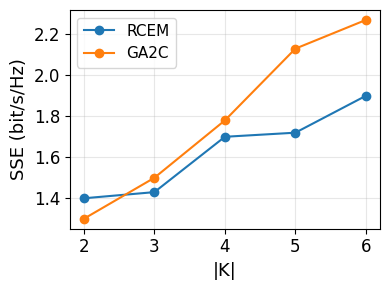}
			\label{fig:sse_vs_k}
		}
		
		\subfloat[]{
			\includegraphics[width=0.45\columnwidth]{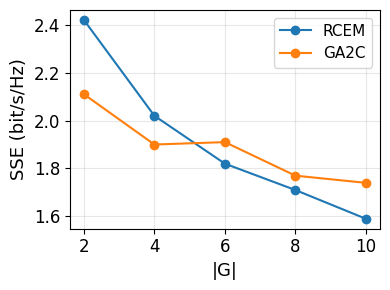}
			\label{fig:sse_vs_g}
		}
		\hfill
		\subfloat[]{
			\includegraphics[width=0.45\columnwidth]{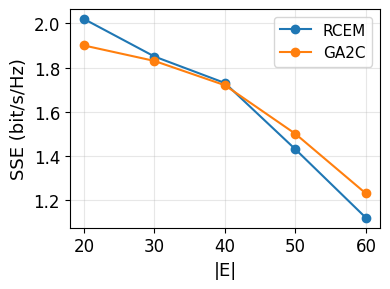}
			\label{fig:sse_vs_e}
		}
		
		\caption{Average per-user SSE using \eqref{eq:sse_average}
			under non-colluding eavesdropping versus
			(a) $|\mathcal{U}|$, (b) $|\mathcal{K}|$,
			(c) $|\mathcal{G}|$, and (d) $|\mathcal{E}|$.
			Settings: (a) $|\mathcal{K}|=4$, $|\mathcal{G}|=4$,
			$|\mathcal{E}|=50$; (b) $|\mathcal{U}|=16$,
			$|\mathcal{G}|=4$, $|\mathcal{E}|=50$;
			(c) $|\mathcal{U}|=16$, $|\mathcal{K}|=3$,
			$|\mathcal{E}|=20$; (d) $|\mathcal{U}|=16$,
			$|\mathcal{K}|=3$, $|\mathcal{G}|=4$.
			Locations are uniformly random, users are evenly grouped,
			and $U_{\mathrm{perRB}}=64$.}
		\label{fig:sse_combined}
	\end{figure}
	
	\subsection{Performance across Independent Random Seeds}
	\label{subsec:random_seeds}
	
	Across ten independent seeds (Fig.~\ref{fig:seed_comparison}),
	GA2C achieves higher mean SSE and fairness. The methods
	share held-out test conditions within each seed; node
	locations, channels, and algorithmic initializations vary
	across seeds.
	
	\begin{figure}[t]
		\centering
		\captionsetup[subfloat]{font=footnotesize}
		
		\subfloat[]{
			\includegraphics[
			width=0.46\columnwidth
			]{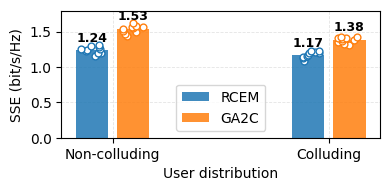}
			\label{fig:seed_sse}
		}
		\hfill
		\subfloat[]{
			\includegraphics[
			width=0.46\columnwidth
			]{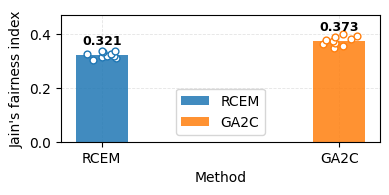}
			\label{fig:seed_fairness}
		}
		
		\caption{Comparison over ten independent seeds:
			(a) average per-user SSE using \eqref{eq:sse_average};
			(b) private-rate fairness using \eqref{eq:jain}.
			Points show seed-level test means; bars average the ten seeds.}
		\label{fig:seed_comparison}
	\end{figure}

	\subsection{Performance under HAPS Outages}
	\label{subsec:haps_outages}
	
	We reoptimize the remaining network after selected HAPSs become unavailable, measuring coverage and retained private-data SSE. The initial configuration has $|\mathcal{U}|=100$, $|\mathcal{K}|=4$, $|\mathcal{G}|=3$, $|\mathcal{E}|=20$, $M=256$, and $U_{\mathrm{perRB}}=64$. Results are averaged first over tested failure sets of equal
	cardinality within each seed, then as percentages over
	$N_{\mathrm{seed}}=10$ independent seeds.
	For a tested outage count $j$, let
	$\mathcal{K}^{(s,j)}$ denote the operational HAPS set,
	with $|\mathcal{K}^{(s,j)}|=4-j$.
	For failure set $F$, the operational HAPS set is
	$\mathcal K^{(s,0)}\setminus F$, whose cardinality is $4-j$
	when $|F|=j$.
	The case $j=0$ denotes the corresponding pre-outage
	configuration.
	Let $a_{u,k}^{\mathsf A,s,F}$ denote the association
	used by method $\mathsf A$ after the failure of HAPS set
	$F$ under seed $s$.
	The coverage ratio is
	\begin{equation}
		\begin{aligned}
			\xi_{\mathrm{cov}}^{\mathsf A,s}(j)
			={}&
			\frac{1}{|\mathcal F_j^{(s)}|\,|\mathcal U|}
			\sum_{F\in\mathcal F_j^{(s)}}
			\sum_{u\in\mathcal U}
			\mathbb I\!\left\{
			\sum_{k\in\mathcal K^{(s,0)}\setminus F}
			a_{u,k}^{\mathsf A,s,F}>0
			\right\}.
		\end{aligned}
		\label{eq:outage_seed_coverage}
	\end{equation}
	
	Accordingly, the mean coverage percentage is
	\begin{equation}
		\overline{\xi}_{\mathrm{cov}}^{\mathsf{A}}(j)
		=
		\frac{100}{N_{\mathrm{seed}}}
		\sum_{s=1}^{N_{\mathrm{seed}}}
		\xi_{\mathrm{cov}}^{\mathsf{A},s}(j).
		\label{eq:outage_mean_coverage}
	\end{equation}
	
	Both panels use the unchanged original set of $100$ users.
	Post-outage optimization omits C8 to expose coverage
	degradation and retains all other constraints.
	
	\begin{figure}[t]
		\centering
		\captionsetup[subfloat]{font=footnotesize}
		
		\subfloat[]{
			\includegraphics[
			width=0.46\columnwidth
			]{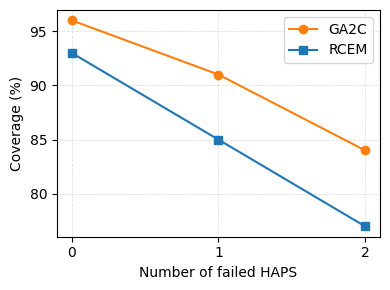}
			\label{fig:haps_coverage}
		}
		\hfill
		\subfloat[]{
			\includegraphics[
			width=0.46\columnwidth
			]{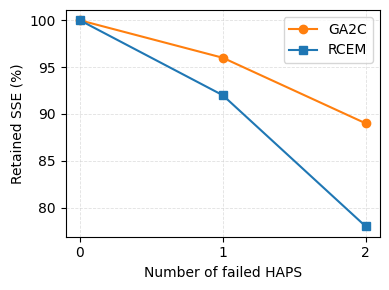}
			\label{fig:haps_retained_sse}
		}
		
		\caption{HAPS-outage performance over 10 independent
			seeds: (a) coverage percentage using
			\eqref{eq:outage_mean_coverage}; (b) retained average per-user
			SSE percentage using \eqref{eq:outage_mean_retention},
			with seed-level averages and ratios in
			\eqref{eq:outage_seed_sse}--\eqref{eq:outage_seed_retention}.
			Initially: 4 HAPSs, 100 users, 3 groups, and 20 eavesdroppers,
			with $U_{\mathrm{perRB}}=64$.}
		\label{fig:haps_outage}
	\end{figure}

	\section{Conclusion}
	\label{sec:conclusion}
	
	This paper proposed SAFA-MZ, which integrates group-tag authentication and private-data communication in NTNs. Its objective balances layout regularity, group-association distribution, and private-rate fairness. Under colluding and non-colluding eavesdropping, RCEM optimized individual snapshots while GA2C learned a reusable policy. GA2C improved average SSE by $135$--$139\%$ over single-connect and by $21.1$--$23.2\%$ over the baseline without AN. Compared with RCEM, it also achieved $23.4\%$ and $17.9\%$ higher mean per-user SSE in non-colluding and colluding scenarios, respectively, and a $16.2\%$ higher mean private-rate Jain's index, demonstrating the benefits of SAFA-MZ.

	\appendices

	\section{Proofs}

	\subsection{Conditional Analytical SSE}
	\label{app:sse_proof}
	
	Condition on the decisions $\boldsymbol\Omega$, legitimate channels,
	and large-scale gains and LoS/NLoS states, denoted by $\mathcal C$.
	Decisions are fixed before independent test eavesdropper fading is drawn.
	Assume mutually independent eavesdropper channels
	$\mathbf h_e=\operatorname{col}_k(\mathbf h_{k,e})
	\sim\mathcal{CN}(\boldsymbol\mu_e,\mathbf S_e^2)$, where
	$\mathbf S_e=\operatorname{blkdiag}_k(\varsigma_{k,e}\mathbf I_M)$.
	From \eqref{eq:rician}, a LoS block has
	$\boldsymbol\mu_{k,e}=\sqrt{\beta_{k,e}\kappa_r/(\kappa_r+1)}
	\mathbf h_{k,e}^{\mathrm{LoS}}$ and
	$\varsigma_{k,e}^2=\beta_{k,e}/(\kappa_r+1)$;
	a NLoS block has zero mean and variance $\beta_{k,e}$.
	For $D_b=M|\mathcal K|$, define $D_b$-dimensional beam vectors by
	\begin{equation}
		\begin{aligned}
			\relax[\mathbf v_u]_k
			&=\sqrt{P_{\mathrm t,k}}\rho_{\mathrm b,k}\mathbf p_{k,u},
			&&k\in\mathcal K_u,\\
			[\mathbf t_g]_k
			&=\sqrt{P_{\mathrm t,k}}\rho_{\mathrm c,k,g}
			\mathbf w_{\mathrm c,k,g},
			&&k\in\mathcal K_g,
		\end{aligned}
	\end{equation}
	with zero inactive blocks. Set
	$\mathbf Q_g=\operatorname{blkdiag}_k
	(\mathbb I\{k\in\mathcal K_g\}P_{\mathrm t,k}\rho_{\mathrm d,k}^2
	\mathbf Q_k^{\mathrm{AN}})$ and, for $g=g(u)$,
	\begin{equation}
		\begin{aligned}
			&\mathbf B_u
			=\sum_{u'\in\mathcal U_g^{\mathrm{srv}}\setminus\{u\}}
			\mathbf v_{u'}\mathbf v_{u'}^H
			+\mathbf t_g\mathbf t_g^H+\mathbf Q_g,\\
			&\mathbf M_{e,u}(x)
			=\mathbf S_e
			(\mathbf v_u\mathbf v_u^H-x\mathbf B_u)\mathbf S_e.
		\end{aligned}
	\end{equation}
	Write
	$\mathbf M_{e,u}(x)=\mathbf U\operatorname{diag}(\lambda_i)\mathbf U^H$
	with unitary $\mathbf U$ and set
	$\boldsymbol\delta=\mathbf U^H\mathbf S_e^{-1}\boldsymbol\mu_e$;
	dependence on $(e,u,x)$ is implicit. 
	Here, \(j\triangleq\sqrt{-1}\) denotes the imaginary unit, and \(\omega\in\mathbb{R}\) is the real-valued Fourier integration variable used in the characteristic function.
	For $x>0$, the quadratic-form
	characteristic function and
	$F_{e,u}(x\mid\mathcal C)=\Pr(\gamma_{e,u}\le x\mid\mathcal C)$ are
	\begin{align}
		\phi_{e,u,x}(\omega)
		&=\frac{\displaystyle\exp\!\left(
			\sum_{i=1}^{D_b}
			\frac{j\omega\lambda_i|\delta_i|^2}{1-j\omega\lambda_i}
			\right)}
		{\displaystyle\prod_{i=1}^{D_b}(1-j\omega\lambda_i)},
		\label{eq:sse_cf_exact}\\
		F_{e,u}(x\mid\mathcal C)
		&=\frac12-\frac1\pi\int_0^\infty
		\frac{\operatorname{Im}\{e^{-j\omega x\sigma^2}
			\phi_{e,u,x}(\omega)\}}{\omega}\,d\omega.
		\label{eq:sse_cdf_exact}
	\end{align}
	Consequently, the conditional mean of \eqref{eq:sse_average} is
	\begin{equation}
		\mathbb E[\overline{\mathrm{SSE}}^{\mathrm{nc}}\mid\mathcal C]
		=\frac{1}{|\mathcal U|\ln2}
		\sum_{u\in\mathcal U^{\mathrm{srv}}}
		\int_0^{\gamma_{\mathrm p,u}}
		\frac{\prod_{e\in\mathcal E}F_{e,u}(x\mid\mathcal C)}{1+x}\,dx.
		\label{eq:sse_conditional_exact}
	\end{equation}
	For \eqref{eq:sse_tag_average}, replace the desired beam by
	$\mathbf t_g$, use interference
	$\sum_{u\in\mathcal U_g^{\mathrm{srv}}}
	\mathbf v_u\mathbf v_u^H+\mathbf Q_g$,
	replace the integration limit by $2^{R_{\mathrm c,g}}-1$,
	and average over groups; inactive groups contribute zero.
	
	\begin{proof}
		Equation~\eqref{eq:eavesdropper_SINR} gives
		$\{\gamma_{e,u}\le x\}=\{\mathbf h_e^H
		(\mathbf v_u\mathbf v_u^H-x\mathbf B_u)\mathbf h_e
		\le x\sigma^2\}$.
		Whitening and unitary diagonalization give independent terms
		$\lambda_i|z_i+\delta_i|^2$, $z_i\sim\mathcal{CN}(0,1)$.
		Their Gaussian integrals give \eqref{eq:sse_cf_exact};
		Gil--Pelaez inversion gives \eqref{eq:sse_cdf_exact}
		\cite{gilpelaez1951inversion}.
		For $Z=\max_e\gamma_{e,u}$,
		$[\log_2(1+\gamma_{\mathrm p,u})-\log_2(1+Z)]^+
		=\frac1{\ln2}\int_0^{\gamma_{\mathrm p,u}}
		\frac{\mathbb I\{Z\le x\}}{1+x}\,dx$.
		Taking expectations and using conditional independence proves
		\eqref{eq:sse_conditional_exact}.
	\end{proof}
	
	Since RZF/MRT beams lie in $\operatorname{range}(\mathbf H_k^H)$,
	$\mathbf P_{\perp,k}$ annihilates them; the scalar blocks of
	$\mathbf S_e$ preserve this orthogonality.
	Hence, the eigenproblem separates into a beam block of dimension
	at most $|\mathcal U_g^{\mathrm{srv}}|+1$ and AN eigenvalues
	$-x\varsigma_{k,e}^2P_{\mathrm t,k}\rho_{\mathrm d,k}^2/\nu_k$,
	each repeated $\nu_k$ times for active $k$ with $\nu_k>0$;
	remaining eigenvalues are zero.

	\subsection{Layout Regularizer}
	\label{app:layout_proof}
	
	\begin{proof}
		Expanding the Laplacian \(\mathbf{L} = \mathbf{D} - \mathbf{W}\) and substituting \([\mathbf{D}]_{k,k} = \sum_{j} [\mathbf{W}]_{k,j}\) yields
		$
		J_{\mathrm{l}} = \frac{1}{2} \sum_{k} \sum_{j} [\mathbf{W}]_{k,j} \|\mathbf{q}_k - \mathbf{q}_j\|^2
		$.
		Since the edge weights increase with the shared-user
		coupling in~\eqref{eq:W}, minimizing $J_{\mathrm l}$
		penalizes large distances between strongly coupled HAPSs.
	\end{proof}
	
	\subsection{Permutation Equivariance}
	\label{app:perm_priv}

	Consider the graph encoder in~\eqref{eq:GNN_Enc},
	$
	\mathbf{H}_{t}^{(\ell+1)}
	=
	\sigma_f\!\left(
	\hat{\mathbf{A}}_t\,
	\mathbf{H}_{t}^{(\ell)}\,
	\boldsymbol{\Theta}^{(\ell)}
	\right)$, where
	$
	\mathbf{H}_{t}^{(0)}=\mathbf{X}_t
	$
	and $\sigma_f(\cdot)$ is element-wise. Let $\mathbf{\Pi}$ be any permutation matrix and define
	$\mathbf{X}'_t=\mathbf{\Pi}\mathbf{X}_t$ and $\hat{\mathbf{A}}'_t=\mathbf{\Pi}\hat{\mathbf{A}}_t\mathbf{\Pi}^{\mathsf{T}}$. Moreover, mean-pooling
	$\mathbf{z}_{\mathrm{gr}}=\frac{1}{|\mathcal{K}|}\sum_{k\in\mathcal{K}}\mathbf{H}_{t}^{(L)}(k,:)$
	is permutation-invariant.
	
	Therefore, the encoder in \eqref{eq:GNN_Enc} is permutation-equivariant, meaning
	$
	\mathbf{H}_{t}^{(\ell)}(\mathbf{X}'_t,\hat{\mathbf{A}}'_t)
	=
	\mathbf{\Pi}\,
	\mathbf{H}_{t}^{(\ell)}(\mathbf{X}_t,\hat{\mathbf{A}}_t)
	$, where $\ell\ge 0$.

	\begin{proof}
		By induction on $\ell$. The base case holds since $\mathbf{H}_{t}^{(0)}=\mathbf{X}_t$.
		Assume the claim for $\ell$; then
		$\hat{\mathbf{A}}'_t\mathbf{H}_t'^{(\ell)}
		=\mathbf{\Pi}\hat{\mathbf{A}}_t\mathbf{H}_t^{(\ell)}$
		and $\sigma_f(\mathbf{\Pi}\mathbf{Y})=\mathbf{\Pi}\sigma_f(\mathbf{Y})$, yielding the $\ell\!+\!1$ case.
		For mean pooling, invariance follows from $\mathbf{1}^{\mathsf T}\mathbf{\Pi}=\mathbf{1}^{\mathsf T}$.
	\end{proof}

	\IfFileExists{refs.bib}{
		\bibliographystyle{IEEEtran}
		\bibliography{refs}
	}{
		\input{references.tex}
	}

\end{document}